%% file: square-colouring-delta2-17-article4.tex
\documentclass[11pt]{article}

\usepackage{graphicx} 
\usepackage[T1]{fontenc}
\usepackage[utf8]{inputenc}
\usepackage{amsmath, amssymb, amsthm}
\usepackage[english]{babel}
\usepackage{eurosym}
\usepackage{fullpage}
\usepackage{palatino}
\usepackage{tikz}
\usepackage{verbatim}
\usepackage{fancybox}

\usepackage[hang,center]{subfigure}
\usepackage{enumerate}\usepackage{epsfig}
\usepackage{amsfonts} 
\usepackage{amsmath}
\usepackage[mathlines]{lineno}
\usepackage{xspace}
\usepackage{tikz}
\usepackage[active]{srcltx}

\usepackage{xspace}
\newcommand{\n}{($N_1$)\xspace}
\newcommand{\nn}{($N_2$)\xspace}
\newcommand{\nnn}{($N_3$)\xspace}
\newcommand{\us}{($S_1$)\xspace}
\newcommand{\uss}{($S_2$)\xspace}
\newcommand{\usss}{($S_3$)\xspace}

\newcommand{\configtikz}{
\tikzstyle{whitenode}=[draw,circle,fill=white,minimum size=8pt,inner sep=0pt]
\tikzstyle{blacknode}=[draw,circle,fill=black,minimum size=8pt,inner sep=0pt]
\tikzstyle{tnode}=[draw,ellipse,fill=white,minimum size=8pt,inner sep=0pt]
\tikzstyle{texte} =[fill=white, text=black]
}

\renewcommand{\neg}[1]{\overline{#1}}

\usetikzlibrary{arrows,shapes,snakes,automata,backgrounds,petri,patterns}

\DeclareMathOperator{\D}{\Delta(G)}

\DeclareMathOperator{\mad}{mad}
\DeclareMathOperator{\ad}{ad}

\newtheorem{lemma}{Lemma}

\newtheorem{theorem}{Theorem}
\newtheorem{claim}{Claim}

\newtheorem{conjecture}{Conjecture}

\def\claimb{$$\vcenter\bgroup\advance\hsize by -8em\noindent
\refstepcounter{claimb}\ignorespaces\it}        
\makeatletter
\def\endclaimb{\rm\egroup\leqno(\theclaim)$$\global\@ignoretrue}
\makeatother

    {\noindent \emph{Proof.} {}{#1}{}}{\hfill
    $\Diamond$\vspace{1em}}

\begin{document}

\title{Graphs with maximum degree $\Delta\geq 17$ and maximum average
  degree less than $3$ are list $2$-distance
  $(\Delta+2)$-colorable\thanks{This work was partially supported by
    the ANR grant EGOS 12 JS02 002 01}}

\author{Marthe Bonamy,
Benjamin Lévêque, Alexandre Pinlou\thanks{Second affiliation: D\'epartement MIAp, Universit\'e Paul-Val\'ery, Montpellier 3}\\ \normalsize{LIRMM, Universit\'e Montpellier 2, CNRS}\\ \small{\{marthe.bonamy, benjamin.leveque, alexandre.pinlou\}@lirmm.fr}}

\maketitle

\begin{abstract} For graphs of bounded maximum average degree, we
  consider the problem of \emph{2-distance coloring}. This is the
  problem of coloring the vertices while ensuring that two vertices
  that are adjacent or have a common neighbor receive different
  colors.  It is already known that planar graphs of girth at least
  $6$ and of maximum degree $\Delta$ are list 2-distance
  $(\Delta+2)$-colorable when $\Delta \geq 24$ (Borodin and Ivanova
  (2009)) and 2-distance $(\Delta+2)$-colorable when $\Delta \geq 18$
  (Borodin and Ivanova (2009)). We prove here that $\Delta \geq 17$
  suffices in both cases. More generally, we show that graphs with
  maximum average degree less than $3$ and $\Delta \geq 17$ are list
  2-distance $(\Delta+2)$-colorable. The proof can be transposed to
  list injective $(\Delta+1)$-coloring.
\end{abstract}

\section{Introduction}

In this paper, we consider only simple and finite graphs. A
\emph{$2$-distance k-coloring} of a graph $G$ is a coloring of the
vertices of $G$ with $k$ colors such that two vertices that are
adjacent or have a common neighbor receive distinct colors. We define
$\chi^2(G)$ as the smallest $k$ such that $G$ admits a $2$-distance
$k$-coloring. This is equivalent to a proper vertex-coloring of the
square of $G$, which is defined as a graph with the same set of
vertices as $G$, where two vertices are adjacent if and only if they
are adjacent or have a common neighbor in $G$. For example, the cycle
of length $5$ cannot be $2$-distance colored with less than $5$ colors
as any two vertices are either adjacent or have a common neighbor:
indeed, its square is the clique of size $5$. An extension of the
2-distance $k$-coloring is the \emph{list 2-distance $k$-coloring},
where instead of having the same set of $k$ colors for the whole
graph, every vertex is assigned some set of $k$ colors and has to be
colored from it. We define $\chi^2_\ell(G)$ as the smallest $k$ such
that $G$ admits a list $2$-distance $k$-coloring of $G$ for any list
assignment. Obviously, $2$-distance coloring is a sub-case of list
$2$-distance coloring (where the same color list is assigned to every
vertex), so for any graph $G$, $\chi^2_\ell(G) \geq
\chi^2(G)$. Kostochka and Woodall~\cite{kw01} even conjectured that it
is actually an equality. The conjecture is still open.

The study of $\chi^2(G)$ on planar graphs was initiated by Wegner in
1977~\cite{w77}, and has been actively studied because of the
conjecture below. The \emph{maximum degree} of a graph $G$ is
denoted $\D$.

\begin{conjecture}[Wegner~\cite{w77}]\label{conj:w77} If $G$ is a
planar graph, then:
\begin{itemize}
\item $\chi^2(G) \leq 7$ if $\D=3$
\item $\chi^2(G) \leq \D+5$ if $4 \leq \D \leq 7$
\item $\chi^2(G) \leq \lfloor \frac{3 \D}{2} \rfloor + 1$ if $\D \geq
8$
\end{itemize}
\end{conjecture}

This conjecture remains open.


Note that any graph $G$ satisfies $\chi^2(G) \geq \D+1$. Indeed, if we
consider a vertex of maximal degree and its neighbors, they form a set
of $\D+1$ vertices, any two of which are adjacent or have a common
neighbor. Hence at least $\D+1$ colors are needed for a $2$-distance
coloring of $G$. It is therefore natural to ask when this lower bound
is reached. For that purpose, we can study, as suggested by Wang and
Lih~\cite{wl03}, what conditions on the sparseness of the graph can be
sufficient to ensure the equality holds.  

A first measure of the sparseness of a planar graph is its girth.  The
\emph{girth} of a graph $G$, denoted $g(G)$, is the length of a
shortest cycle.  Wang and Lih~\cite{wl03} conjectured that for any
integer $k \geq 5$, there exists an integer $D(k)$ such that for every
planar graph $G$ verifying $g(G) \geq k$ and $\D \geq D(k)$,
$\chi^2(G)=\D+1$.  This was proved by Borodin, Ivanova and
Noestroeva~\cite{bin04,bin08} to be true for $k \geq 7$, even in the
case of list-coloring, and false for $k\in \{5,6 \}$. So far, in the
case of list coloring, it is known~\cite{blp12,i11} that we can choose
$D(7)= 16$, $D(8)= 10$, $D(9) = 8$, $D(10)= 6$, $D(12)= 5$.
Borodin, Ivanova and Neustroeva~\cite{bin06} proved that the case
$k=6$ is true on a restricted class of graphs, i.e.  for a planar
graph $G$ with girth $6$ where every edge is incident to a vertex of
degree at most two and $\D\geq 179$, we have $\chi^2(G)\leq \D+1$.
Dvo\v{r}\'{a}k et al.~\cite{dkns08} proved that the case $k=6$ is true
by allowing one more color, i.e.  for a planar graph $G$ with girth
$6$ and $\D\geq 8821$, we have $\chi^2(G)\leq \D+2$. They also
conjectured that the same holds for a planar graph $G$ with girth $5$
and sufficiently large $\D$, but this remains open. Borodin and
Ivanova improved~\cite{bi08} Dvo\v{r}\'{a}k et al.'s result and
extended it to list-coloring~\cite{bi08b,bi09} as follows.

\begin{theorem}[Borodin and Ivanova~\cite{bi08}]\label{thm:bi08} Every
planar graph $G$ with $\Delta(G) \geq 18$ and $g(G)\geq 6$ admits a
$2$-distance $(\Delta(G)+2)$-coloring.
\end{theorem}

\begin{theorem}[Borodin and Ivanova~\cite{bi09}]\label{thm:bi09} Every
planar graph $G$ with $\Delta(G) \geq 24$ and $g(G)\geq 6$ admits a
list $2$-distance $(\Delta(G)+2)$-coloring.
\end{theorem}

Theorems~\ref{thm:bi08} and~\ref{thm:bi09} are optimal with regards to
the number of colors, as shown by the family of graphs presented by
Borodin et al.~\cite{bgint04}, which are of increasing maximum degree,
of girth $6$ and are not $2$-distance $(\Delta+1)$-colorable.  We
improve Theorems~\ref{thm:bi08} and~\ref{thm:bi09} as follows.

\begin{theorem}\label{cor:m6} Every planar graph $G$ with $\Delta(G)
\geq 17$ and $g(G)\geq 6$ admits a list $2$-distance
$(\Delta(G)+2)$-coloring.
\end{theorem}

Another way to measure the sparseness of a graph is through its
maximum average degree. The \emph{average degree} of a graph $G$,
denoted $\ad(G)$, is $\frac{\sum_{v \in
    V}d(v)}{|V|}=\frac{2|E|}{|V|}$. The \emph{maximum average degree}
of a graph $G$, denoted $\mad(G)$, is the maximum of $\ad(H)$ over all
subgraphs $H$ of $G$. Intuitively, this measures the sparseness of a
graph because it states how great the concentration of edges in a same
area can be. For example, stating that $\mad(G)$ has to be smaller
than $2$ means that $G$ is a forest. Using this measure, we prove a
more general theorem than Theorem~\ref{cor:m6}.

\begin{theorem}\label{thm:main} Every graph $G$ with $\Delta(G) \geq
17$ and $\mad(G)<3$ admits a list $2$-distance
$(\Delta(G)+2)$-coloring.
\end{theorem}

Euler's formula links girth and maximum average degree in the case of
planar graphs.

\begin{lemma}[Folklore]\label{lem:euler} For every planar graph $G$,
$(\mad(G)-2)(g(G)-2)<4$.
\end{lemma}

By Lemma~\ref{lem:euler}, Theorem~\ref{thm:main} implies
Theorem~\ref{cor:m6}.

An \emph{injective k-coloring}~\cite{hkss02} of $G$ is a (not
necessarily proper) coloring of the vertices of $G$ with $k$ colors
such that two vertices that have a common neighbor receive distinct
colors. We define $\chi_i(G)$ as the smallest $k$ such that $G$ admits
an injective $k$-coloring.  A $2$-distance $k$-coloring is an
injective $k$-coloring, but the converse is not true. For example, the
cycle of length $5$ can be injective colored with 3 colors. The list
version of this coloring is a \emph{list injective k-coloring} of $G$,
and $\chi_{i,\ell}(G)$ is the smallest $k$ such that $G$ admits a list
injective $k$-coloring.

Some results on 2-distance coloring have their counterpart on
injective coloring with one less color. This is the case of
Theorems~\ref{thm:bi08} and~\ref{thm:bi09}~\cite{bi10,bi11}. The proof
of Theorem~\ref{thm:main} also works with close to no alteration for
list injective coloring, thus yielding a proof that every graph $G$
with $\Delta(G) \geq 17$ and $\mad(G)<3$ admits a list injective
$(\Delta(G)+1)$-coloring.

In Sections~\ref{sect:def} and~\ref{sect:term}, we introduce the method
and terminology. In Sections~\ref{sect:conf} and~\ref{sect:dis}, we
prove Theorem~\ref{thm:main} and its counterpart on injective coloring
by a discharging method.

\section{Method}\label{sect:def}

The discharging method was introduced in the beginning of the
20$^{th}$ century. It has been used to prove the celebrated Four Color
Theorem in \cite{ah77,ahk77}.  A discharging method is said to be
\emph{local} when the weight cannot travel arbitrarily far. Borodin,
Ivanova and Kostochka introduced in~\cite{bik05} the notion of
\emph{global} discharging method, where the weight can travel
arbitrarily far along the graph. 


We prove for induction purposes a slightly stronger version of
Theorem~\ref{thm:main} by relaxing the constraint on the maximum
degree. Namely, we relax it into ``For any $k \geq 17$, every graph
$G$ with $\Delta(G)\leq k$ and $\mad(G)<3$ verifies
$\chi^2_\ell(G)\leq k+2$'' so that the property is closed under
vertex- or edge-deletion.  A graph is \emph{minimal} for a property if
it satisfies this property but none of its subgraphs does.

The first step is to consider a minimal counter-example $G$, and prove
it cannot contain some configurations. To do so, we assume by
contradiction that $G$ contains one of the configurations. We consider
a particular subgraph $H$ of $G$, and color it by minimality (the
maximum average degree of any subgraph of $G$ is bounded by the
maximum average degree of $G$). We show how to extend the coloring of
$H$ to $G$, a contradiction.

The second step is to prove that a graph that does not contain any of
these configurations has a maximum average degree of at least $3$.  To
that purpose, we assign to each vertex its degree as a weight. We
apply discharging rules to redistribute weights along the graph with
conservation of the total weight.  As some configurations are
forbidden, we can then prove that after application of the discharging
rules, every vertex has a final weight of at least $3$.  This implies
that the average degree of the graph is at least $3$, hence the
maximum average degree is at least $3$. So a minimal counter-example
cannot exist.

We finally explain how the same proof holds also for list injective $(\Delta+1)$-coloring.

\section{Terminology}\label{sect:term}

In the figures, we draw in black a vertex that has no other neighbor
than the ones already represented, in white a vertex that might have
other neighbors than the ones represented. White vertices may coincide
with other vertices of the figure. When there is a label inside a
white vertex, it is an indication on the number of neighbors it
has. The label '$i$' means "exactly $i$ neighbors", the label '$i^+$'
(resp. '$i^-$') means that it has at least (resp. at most) $i$
neighbors.

Let $u$ be a vertex.  The \emph{neighborhood} $N(u)$ of $u$ is the set
of vertices that are adjacent to $u$.  Let $d(u)=|N(u)|$ be the
\emph{degree} of $u$. A \emph{$p$-link} $x-a_1-...-a_p-y$, $p \geq 0$,
between $x$ and $y$ is a path between $x$ and $y$ such that
$d(a_1)=...=d(a_p)=2$. When a $p$-link exists between two vertices $x$
and $y$, we say they are \emph{$p$-linked}. If there is a $p$-link $x-a_1-...-a_p-y$ between $x$ and $y$, we say $x$ is \emph{$p$-linked through $a_1$} to $y$. A \emph{partial $2$-distance list coloring} of $G$ is a $2$-distance list-coloring of a subgraph $H$ of $G$.

A vertex is \emph{weak} when it is of degree $3$ and is $1$-linked to
two vertices of degree at most $14$, or twice $1$-linked to a vertex of degree at most $14$ (see Figure~\ref{fig:weak}). 
A weak vertex is represented with a $w$ label inside ($\neg w$ if it is not weak).

\begin{figure}[h] \center
\input{weak}
\caption{A weak vertex $x$.}
\label{fig:weak}
\end{figure}
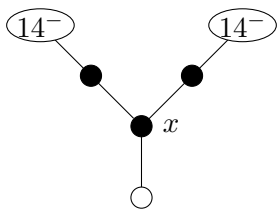

 A vertex is \emph{support} when it is either (see Figure~\ref{fig:neg}): 
\begin{enumerate}[Type ($S_1$): ]
\item a vertex of
degree $2$ adjacent to another vertex of degree $2$;
\item a vertex of degree $2$ that is adjacent to a vertex of degree
  $3$ which is adjacent to a vertex of degree $2$ and to a vertex of
  degree at most $7$;
\item a weak vertex $1$-linked to another weak vertex.
\end{enumerate}

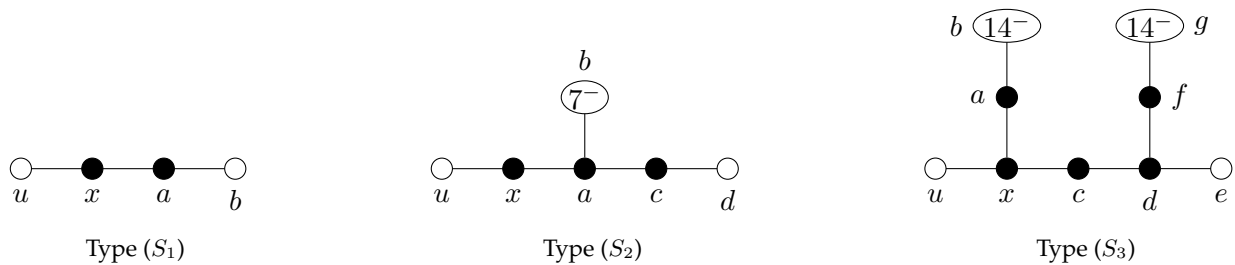
\begin{figure}
\center
\subfigure[Type ($S_1$)]{\input{neg1}} \hfill
\subfigure[Type ($S_2$)]{\input{neg2}} \hfill
\subfigure[Type ($S_3$)]{\input{neg3}}
\caption{Support vertices $x$.}
\label{fig:neg}
\end{figure}

A vertex is \emph{positive} when it is of degree at least $4$ and is
adjacent to a support vertex. A vertex $u$ is \emph{locked} if it has
two neighbors $v_1$ and $v_2$, where $v_1$ and $v_2$ are both
$1$-linked to the same two vertices $w_1$ and $w_2$ that have a common
neighbor, and $d(v_1)=d(v_2)=d(w_1)=d(w_2)=3$ (see
Figure~\ref{fig:lock}). This configuration is called a \emph{lock}.

\begin{figure}[h] \center

\input{lock}
\caption{A locked vertex $u$.}
\label{fig:lock}
\end{figure}
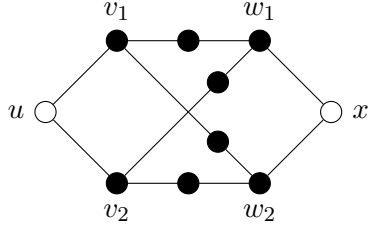

\section{Forbidden Configurations}\label{sect:conf}

In all the paper, $k$ is a constant integer greater than $17$ and $G$
is a minimal graph such that $\Delta(G) \leq k$ and $G$ admits no
$2$-distance $(k+2)$-list-coloring.

We define configurations {($C_1$)} to {($C_{11}$)} (see
Figures~\ref{fig:config1-5},~\ref{fig:config6-9}
  and~\ref{fig:config10-11}). Note that configurations similar to Configurations {($C_1$)}, {($C_2$)} and {($C_4$)} already existed in the litterature, for example in~\cite{dkns08}.
\begin{itemize}
\item {($C_1$)} is a vertex $u$ with $d(u)\leq 1$
\item {($C_2$)} is a vertex $u$ with $d(u)=2$ that has two
  neighbors $v,w$ and $u$ is $1$-linked through $v$ to a vertex of
  degree at most $k-1$.
\item {($C_3$)} is a vertex $u$ with $d(u)=3$ that has three
  neighbors $v,w,x$ with $d(w)+d(x)\leq k-1$, and $u$ is $1$-linked
  through $v$ to a vertex of degree at most $k-1$.
\item {($C_4$)} is a vertex $u$ with $d(u)=3$ that has three
  neighbors $v,w, x$ with $d(w)+d(x) \leq k-1$, and $v$ has exactly three
  neighbors $u, y, z$ with $d(z)\leq 7$ and $d(y)=2$.
\item {($C_5$)} is a vertex $u$ with $d(u)=3$ that has three
  neighbors $v, w, x$ with $d(x)\leq k-1$ and $u$ is
  $1$-linked through $v$ (resp. through $w$) to a vertex of degree
  at most $14$. (Note that $u$ is  weak vertex.)
\item {($C_6$)} is a vertex $u$ with $d(u)=4$ that has four
  neighbors $v, w, x, y$ with $d(w)\leq 7$, $d(x)\leq 3$, $d(y)\leq
  3$, and $u$ is $1$-linked through $v$ to a vertex of degree at most
  $14$.
\item {($C_7$)} is a vertex $u$ with $d(u)=4$ that has four
  neighbors $v, w, x, y$ with $d(x)+d(y)\leq k-1$ and $u$ is
  $1$-linked through $v$ (resp. through $w$) to a vertex of degree at
  most $14$.

\item {($C_8$)} is a vertex $u$ with $d(u)=5$ that has five
  neighbors $v, w, x, y, z$ with $d(w)\leq 7$, $d(x)\leq 3$, $d(y)\leq
  3$, $d(z)=2$, and $u$ is $1$-linked through $v$ to a vertex of
  degree at most $7$.

\item {($C_9$)} is a vertex $u$ with $d(u)=6$ that has six
  neighbors $v, w, x, y, z, t$ with $d(w)\leq 7$, $d(x)\leq 3$,
  $d(y)\leq 3$, $d(z)=2$, $d(t)=2$, and $u$ is $1$-linked through $v$
  to a vertex of degree at most $7$.

\item {($C_{10}$)} is a vertex $u$ with $d(u)=7$ that has seven
  neighbors $v, w_1, \ldots, w_6$ with $d(v)\leq 7$ and $u$ is
  $1$-linked through $w_i$, $1\leq i\leq 6$, to a vertex of degree at
  most $3$.

\item {($C_{11}$)} is a vertex $u$ with $d(u)=k$ that has three
  neighbors $v, w, x$ with $x$ is a support vertex, $v, w$ are both
  $1$-linked to a same vertex $y$ of degree $3$, and $v$ (resp. $w$)
  is $1$-linked to a vertex of degree at most $14$ distinct from
  $y$. (Note that $v, w$ are weak vertices.)
\end{itemize}

\begin{lemma}\label{lem:config} 
  $G$ does not contain  Configurations {($C_1$)} to
  {($C_{11}$)}.
\end{lemma}

\begin{proof}
  Given a partial 2-distance list-coloring of $G$, a \emph{constraint}
  of a vertex $u$ is color appearing on a vertex at distance at most
  $2$ from $u$ in $G$.  

  Notation refers to Figures~\ref{fig:config1-5},~\ref{fig:config6-9}
  and~\ref{fig:config10-11}. 

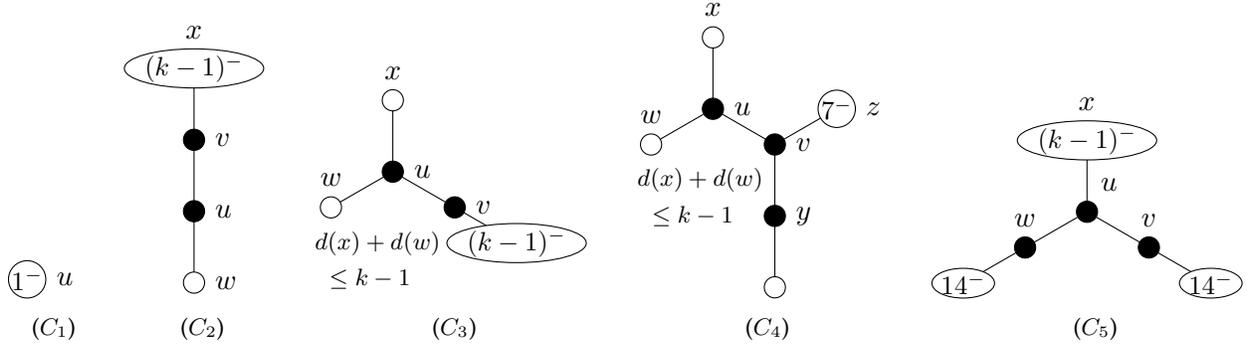
\begin{figure}
  \subfigure[($C_1$)]{\input{c01}} \hfill
  \subfigure[($C_2$)]{\input{c02}} \hfill
  \subfigure[($C_3$)]{\input{c03}} \hfill
  \subfigure[($C_4$)]{\input{c04}} \hfill
  \subfigure[($C_5$)]{\input{c05}} 
  \caption{Forbidden configurations ($C_1$) to ($C_5$).}
  \label{fig:config1-5}
\end{figure}

\begin{claim} $G$ does not contain {($C_1$)}.
\end{claim}
\begin{proof} Suppose by contradiction that $G$ contains ($C_1$). Using
  the minimality of $G$, we color $G \setminus \{u\}$. Since
  $\Delta(G) \leq k$, and $d(u) \leq 1$, vertex $u$ has at most $k$
  constraints (one for its neighbor and at most $k-1$ for the vertices
  at distance $2$ from $u$). There are $k+2$ colors available in the
  list of $u$, so the coloring of $G \setminus \{u\}$ can be extended
  to $G$, a contradiction.
\end{proof}

\begin{claim} $G$ does not contain {($C_2$)}.
\end{claim}
\begin{proof} Suppose by contradiction that $G$ contains ($C_2$). Using
  the minimality of $G$, we color $G \setminus \{u,v\}$. Vertex $u$
  has at most $k+1$ constraints. Hence we can color $u$. Then $v$ has
  at most $k-1+2=k+1$ constraints. Hence we can color $v$. So we can
  extend the coloring to $G$, a contradiction.
\end{proof}

\begin{claim} $G$ does not contain {($C_3$)}.
\end{claim}
\begin{proof} 
  Suppose by contradiction that $G$ contains ($C_3$).  Using the
  minimality of $G$, we color $G \setminus \{v\}$. Because of $u$, vertices $w$
  and $x$ have different colors.  We discolor $u$.  Vertex $v$ has at
  most $k-1+2=k+1$ constraints. Hence we can color $v$.  Vertex $u$
  has at most $d(w)+d(x)+2 \leq k+1$ constraints. Hence we can color
  $u$. So we can extend the coloring to $G$, a contradiction.
\end{proof}

\begin{claim} $G$ does not contain {($C_4$)}.
\end{claim}
\begin{proof} Suppose by contradiction that $G$ contains ($C_4$). Let
  $e$ be the edge $uv$. Using the minimality of $G$, we color $G
  \setminus \{e\}$. We discolor $u$ and $v$. Vertex $u$ has at most
  $d(w)+d(x)+2\leq k+1$ constraints. Hence we can color $u$.  Vertex
  $v$ has at most $7+3+2 \leq k+1$ constraints. Hence we can color
  $v$. So we can extend the coloring to $G$, a contradiction.
\end{proof}

\begin{claim} $G$ does not contain {($C_5$)}.
\end{claim}
\begin{proof} Suppose by contradiction that $G$ contains ($C_5$). Using
  the minimality of $G$, we color $G \setminus \{u,v,w\}$. Vertex $u$
  has at most $k-1+2=k+1$ constraints. Hence we can color $u$.
  Vertices $v$ and $w$ have at most $14+3 \leq k+1$ constraints
  respectively. Hence we can color $v$ and $w$. So we can extend the
  coloring to $G$, a contradiction.
\end{proof}

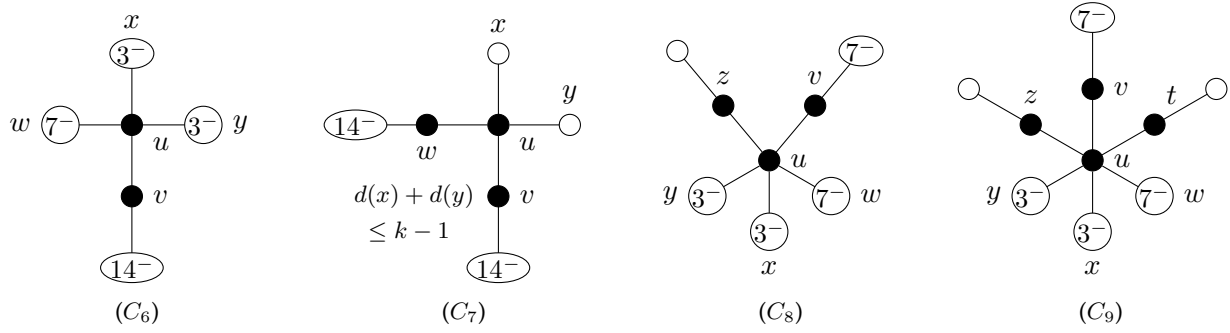
\begin{figure}
  \subfigure[($C_6$)]{\input{c06}} \hfill
  \subfigure[($C_7$)]{\input{c07}} \hfill
  \subfigure[($C_8$)]{\input{c08}} \hfill
  \subfigure[($C_9$)]{\input{c09}} \hfill
  \caption{Forbidden configurations ($C_6$) to ($C_9$).}
  \label{fig:config6-9}
\end{figure}

\begin{claim} $G$ does not contain {($C_6$)}.
\end{claim}
\begin{proof} Suppose by contradiction that $G$ contains ($C_6$). Using
  the minimality of $G$, we color $G \setminus \{v\}$. We discolor
  $u$. Vertex $v$ has at most $14+3 \leq k+1$ constraints. Hence we
  can color $v$. Vertex $u$ has at most $2+3+3+7 \leq k+1$
  constraints. Hence we can color $u$. So we can extend the coloring
  to $G$, a contradiction.
\end{proof}

\begin{claim} $G$ does not contain {($C_7$)}.
\end{claim}
\begin{proof} Suppose by contradiction that $G$ contains ($C_7$). Using
  the minimality of $G$, we color $G \setminus \{v,w\}$. We discolor
  $u$. Vertex $u$ has at most $d(x)+d(y)+2 \leq k+1$
  constraints. Hence we can color $u$. Vertices $v$ and $w$ have at
  most $14+4 \leq k+1$ constraints respectively. Hence we can color
  $v$ and $w$. So we can extend the coloring to $G$, a contradiction.
\end{proof}

\begin{claim} $G$ does not contain {($C_8$)}.
\end{claim}
\begin{proof} Suppose by contradiction that $G$ contains ($C_8$). Using
  the minimality of $G$, we color $G \setminus \{v\}$. We discolor
  $u$. Vertex $u$ has at most $7+3+3+2+1 \leq k+1$ constraints.  Hence
  we can color $u$. Vertex $v$ has at most $7+5 \leq k+1$
  constraints. Hence we can color $v$. So we can extend the coloring
  to $G$, a contradiction.
\end{proof}

\begin{claim} $G$ does not contain {($C_9$)}.
\end{claim}
\begin{proof} Suppose by contradiction that $G$ contains ($C_9$). Using
  the minimality of $G$, we color $G \setminus \{v\}$. We discolor
  $u$. Vertex $u$ has at most $7+3+3+2+2+1 \leq k+1$ constraints.
  Hence we can color $u$. Vertex $v$ has at most $7+6 \leq k+1$
  constraints. Hence we can color $v$. So we can extend the coloring
  to $G$, a contradiction.
\end{proof}

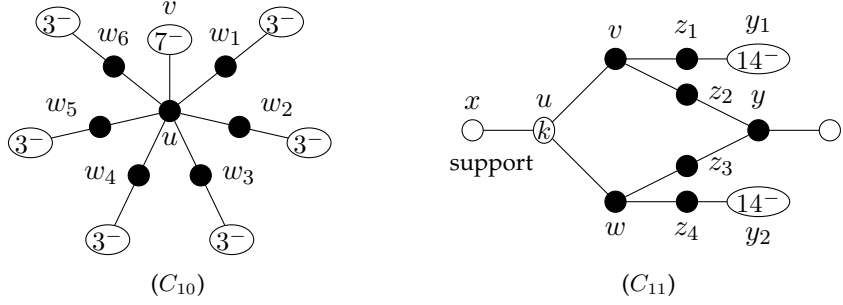
\begin{figure} \centering 
  \subfigure[($C_{10}$)]{\input{c10}} $\qquad\quad$
  \subfigure[($C_{11}$)]{\input{c11}} 
  \caption{Forbidden configurations ($C_{10}$) and ($C_{11}$).}
  \label{fig:config10-11}
\end{figure}

\begin{claim} $G$ does not contain {($C_{10}$)}.
\end{claim}
\begin{proof} Suppose by contradiction that $G$ contains
  ($C_{10}$). Using the minimality of $G$, we color $G \setminus
  \{u,w_1,\ldots,w_6\}$. Vertex $u$ has at most $7+6 \leq k+1$
  constraints. Hence we can color $v$. Vertices $w_i$ have at most
  $3+7 \leq k+1$ constraints. Hence we can color $w_1,...,w_6$. So we
  can extend the coloring to $G$, a contradiction.
\end{proof}

\begin{claim} $G$ does not contain {($C_{11}$)}.
\end{claim}
\begin{proof} Suppose by contradiction that $G$ contains ($C_{11}$).
  Since $x$ is a support vertex, and $u$ is of degree $k$, it is of type
  ($S_1$), ($S_2$) or ($S_3$) of support vertices with the notation
  of Figure~\ref{fig:neg}.  Note that some vertices may coincide
  between Figure~\ref{fig:neg} and Figure~\ref{fig:config10-11}.










We define a set of vertices $A$ as follows:

$$A=\left\{ \begin{array}{lr} 
    \{a\} & \text{if }x\text{ is of Type (}S_{1}\text{)} \\
    \{a,c\} & \text{if }x\text{ is of Type (}S_{2}\text{)} \\
    \{a,c\} & \text{if }x\text{ is of Type  (}S_{3}\text{)} 
 \end{array}  \right. $$

Using the minimality of $G$, we color $G \setminus (\{v,w,x,y,$
$z_1,\ldots,z_4\} \cup A)$.  If $x$ is of Type ($S_1$)
(resp. ($S_2$)), $a$ (resp. $c$) has at most $k+1$ constraints, hence we can color
it.  For the three types ($S_i$), $x$ has at most $k-3+1+2 = k$
constraints, thus it has at least $2$ available colors.  Vertex $y$
has at most $k$ constraints, thus it has at least $2$ available
colors.  Both $v$ and $w$ have at most $k-3+1+1\leq k-1$ constraints,
so they have at least $3$ available colors in their list.

We now explain how to color $v,w,x,y$ (other uncolored vertices will
be colored after).  Suppose $x$ and $y$ can be assigned the same
color, then both $v$ and $w$ have at least $2$ available colors and
thus can be colored.

Suppose the lists of available colors of $x$ and $y$ are disjoint.  We
color $v$ with a color not appearing in the list of $x$. Then we color
$y$ that has $k+1$ constraints. (Vertex $x$ has still at least $2$
available colors.) Then we color $w$ that has $k+1$ constraints and
finally $x$. 

Now we assume that we cannot assign the same color to $x$ and $y$ and
that their lists of available colors are not disjoint. This means that
$x$ and $y$ are either adjacent or have a common neighbor. So some
vertices coincide between Figure~\ref{fig:neg} and
Figure~\ref{fig:config10-11}. The different cases where $x$ and $y$ are either
adjacent or have a common neighbor are the following:

  \begin{itemize}

  \item[($S_1$)] 

    \begin{itemize}
    \item $b=y$
    \end{itemize}

\item[($S_2$)]
  \begin{itemize}
 \item $b=y$
 \item $a=y$ and w.l.o.g $b=z_2$, $c=z_3$ and $d=w$.
 \end{itemize}

\item[($S_3$)] 
  \begin{itemize}
  \item $b=y$
  \item $d=y$, and w.l.o.g. $f=z_2$, $g=v$ and $e=z_3$.  
\end{itemize}
\end{itemize}

In all these cases, $y$ has at most $1$ contraint.  So we can color
$x,v,w,y$, in this order as they all have at most $k+1$ constraints
when they are colored.

If $x$ is of Type ($S_2$) (resp. ($S_3$)), vertex $a$ (resp vertices
$a$, $c$) has at most $11$ constraints (resp. $18$, $6$), so we can
color them. The vertices $z_i$ have at most $17\leq k+1$, so we can
color them. Thus the coloring have been extended to $G$, a
contradiction.
\end{proof}

\end{proof}


\section{Structure of support vertices}
\label{sec:structure}

Let $H(G)$ be the subgraph of $G$ induced by the edges incident
to at least a support vertex. 
We prove several properties of support vertices and of the graph
$H(G)$.

\begin{lemma}\label{lem:supportpositive} 
  Each positive vertex is of degree $k$ and each support vertex is
  adjacent to exactly one positive vertex.
\end{lemma}
\begin{proof}
  By Lemma~\ref{lem:config}, $G$ does not contain Configurations
  {($C_2$)}, {($C_3$)} and {($C_5$)}. So a support vertex is adjacent
  to a vertex of degree $k$ (Configurations {($C_2$)}, {($C_3$)} and
  {($C_5$)} correspond respectively to support vertices of Type
  ($S_1$), ($S_2$) and ($S_3$)).  By definition, a support vertex has
  at most one neighbor of degree at least $4$, thus it is adjacent to
  exactly one vertex of degree at least $4$ and this vertex has in
  fact degree $k$.  So all the positive vertices are of degree $k$ and
  a support vertex is adjacent to exactly one positive vertex.
\end{proof}

\begin{lemma}
  \label{lem:cycleimpair}
  Each cycle of $H(G)$ with an odd number of support vertices contains
  a subpath $s_1v_1s_2v_2s_3$ where $s_1,s_2,s_3$ are support
  vertices of type ($S_3$) and $v_1,v_2$ are vertices of degree 2.
\end{lemma}

\begin{proof}
  Let $C$ be cycle of $H(G)$ with an odd number of support vertices.
  Cycle $C$ does not contain just one support vertex, as all its edges
  have to be adjacent to a support vertex (there is no loop nor
  multiple edge in $H(G)$). So $C$ contains at least three support
  vertices.

  Suppose that $C$ contains no positive vertices. Then it contains no
  support vertices of type \us or \uss as such vertices are of degree
  $2$, so all their neighbors would  be on $C$, and they are adjacent to
  a positive vertex by Lemma~\ref{lem:supportpositive}. So $C$
  contains only support vertices of type \usss.  Let $s_1,s_2,s_3$ be
  three support vertices of $C$ appearing consecutively along $C$.  A
  support vertex of Type \usss is of degree $3$, adjacent to two
  vertices of degree $2$ and to a positive vertex. So the neighbors of
  $s_i$ on $C$ are vertices of degree $2$ that are not support
  vertices. As $H(G)$ contains only edges incident to support
  vertices, there exist $v_1,v_2$ of degree $2$ such that
  $s_1v_1s_2v_2s_3$ is a subpath  of $C$.

  Suppose now that $C$ contains some positive vertices.  Let
  $p_1,\ldots,p_\ell$ be the set of positive vertices of $C$ appearing
  in this order along $C$ while walking in a chosen direction
  (subscript are understood modulo $\ell$).  Let $Q_i$, $1\leq i\leq
  \ell$, be the subpath of $C$ between $p_i$ and $p_{i+1}$ (in the
  same choosen direction along $C$). (Note that if $\ell=1$, then $Q_1=C$
  is not really a subpath.) As $C$ contains an odd number of support
  vertices, there exists $i$ such that $Q_i$ contains an odd number of
  support vertices.  If $Q_i$ contains just one support vertex $v$,
  then $Q_i$ has length $2$, since $H(G)$ contains only edges incident
  to support vertices. So $v$ is adjacent to two different positive
  vertices (or has a multiple edge if $\ell=1$), a contradiction to
  Lemma~\ref{lem:supportpositive}. So $Q_i$ contains at least $3$
  support vertices. Let $s_1,s_2,s_3$ be three support vertices of
  $Q_i$ appearing consecutively along $Q_i$.

  If one of the $s_i$ is of Type \us, let $x$ be such a vertex. With
  the notation of Figure~\ref{fig:neg}, vertex $x$ is of degree $2$,
  so its two neighbors $u,a$ are on $C$, with $u$ a positive vertex
  and $a$ a support vertex of Type \us. Then vertex $a$ is of degree
  $2$ so its neighbor $b$ distinct from $x$ is also on $C$. Vertex $b$
  is positive so $Q_i$ is the path $u,x,a,b$ and contains just two
  support vertices, a contradiction.

  If one of the $s_i$ is of Type \uss, let $x$ be such a vertex. With
  the notation of Figure~\ref{fig:neg}, vertex $x$ is of degree $2$,
  so its two neighbors $u,a$ are on $C$, with $u$ a positive vertex
  and $a$ a vertex of degree $3$. Vertex $a$ is not adjacent to
  vertices of degree $k$ so by Lemma~\ref{lem:supportpositive}, it is
  not a support vertex. Let $c'$ be the neighbor of $a$ on $C$ that is
  distinct from $x$. As all the edges of $H(G)$ are incident to
  support vertices, $c'$ is a support vertex. Since $c'$ is adjacent
  to a vertex of degree $3$ it is a support vertex of Type \uss and
  can play the role of $c$ of Figure~\ref{fig:neg}. Then $c$ is of
  degree $2$ and its neighbor on $C$ distinct from $a$ is a positive
  vertex $d$. So $Q_i$ is the path $u,x,a,c,d$ and contains just two
  support vertices, a contradiction.

  So $s_1,s_2,s_3$ are all of Type \usss.  A support vertex of Type
  \usss is of degree $3$, adjacent to two vertices of degree $2$ and
  to a positive vertex. So the neighbors of $s_2$ on $C$ are vertices
  $v_1,v_2$ of degree $2$ that are not support vertices. As $H(G)$
  contains only edges incident to support vertices, we can assume
  w.l.o.g. that $s_1v_1s_2v_2s_3$ is a subpath of $C$.
\end{proof}

  \begin{lemma}
    \label{cl:cycle2}
    $H(G)$ does not contain a 2-connected subgraph of size at least
    three with exactly two support vertices.
  \end{lemma}
  \begin{proof}
    Suppose by contradiction that $H(G)$ contains a 2-connected
    subgraph $C$ of size $\geq 3$ that has exactly two support
    vertices $S=\{s_1,s_2\}$.  We color by minimality $G\setminus
    (S\cup\{v \in N_G(S) | d_G(v)\leq 3\})$. (Note that by
    Lemma~\ref{lem:supportpositive}, the set $\{v \in N_G(S) |
    d_G(v)\leq 3\}$ corresponds to vertex $a$ of Figure~\ref{fig:neg}
    if the support vertex is of Type \us or \uss and to vertices
    $a,c$ if the support vertex is of Type ($S_3$).)  

    We first show how to color $S$.  For that purpose we consider
    three cases corresponding to the type of $s_1$.

  \begin{itemize}
  \item \emph{$s_1$ is of Type ($S_1$).} Then $s_1$ is of degree $2$,
    has a positive neighbor $u$ and a support neighbor $a$ of Type
    ($S_1$).  As $s_1$ is of degree $2$, both its neighbors are in
    $C$. So $a$ is a support vertex of $C$, thus $a=s_2$.
    Then $u$ is of degree $k$, has two neighbors $s_1,s_2$ that are
    not colored, so $s_1$ and $s_2$ have at most $k$ constraints, and
    we can color them.

  \item \emph{$s_1$ is of Type ($S_2$).} Then $s_1$ is of degree $2$,
    has a positive neighbor $u$ and another neighbor $a$ of degree
    $3$. Vertex $a$ is not a support vertex by
    Lemma~\ref{lem:supportpositive} since it has no neighbor of degree
    $k$.  As $s_1$ is of degree $2$, all its neighbors are in $C$.
    Vertices $u$ and $a$ are in $C$ that is 2-connected so they have
    at least two neighbors in $C$. Since they are not support
    vertices, all their neighbors in $C$ are support
    vertices. So both $u$ and $a$ are adjacent to $s_2$.  Vertex $s_2$
    is support, it is adjacent to $a$ that is of degree $3$, so $s_2$
    is of Type ($S_2$).  Then $u$ is of degree $k$, has two neighbors
    $s_1,s_2$ that are not colored, so $s_1$ and $s_2$ have at most
    $k$ constraints, and we can color them.

  \item \emph{$s_1$ is of Type ($S_3$).} Then $s_1$ is of degree $3$,
    has a positive neighbor $u$ and two other neighbors $w,w'$ of
    degree $2$. Vertices $w,w'$ are not support vertices by
    Lemma~\ref{lem:supportpositive} since they have no neighbor of
    degree $k$.  As $s_1$ is of degree $3$, two of $u,w,w'$ are in
    $C$. Let $Y$ be the neighbors of $s_1$ in $C$. We can assume by
    symmetry that either $\{v,w\}\subseteq Y$ or $\{w,w'\}\subseteq
    Y$.  Vertices of $Y$ are in $C$ that is 2-connected so they have
    at least two neighbors in $C$. Since they are not support
    vertices, all their neighbors in $C$ are support vertices. So all
    the vertices of $Y$ are adjacent to $s_2$.  Vertex $s_2$ is a
    support vertex, it is adjacent to $w$ that is non support and of
    degree $2$, so $s_2$ is of Type ($S_3$). In both cases
    ($\{v,w\}\subseteq Y$ or $\{w,w'\}\subseteq Y$), vertices $s_1$
    and $s_2$ have at most $k$ constraints, and we can color them.
\end{itemize}

 Every vertex of $\{v \in N_G(S) | d_G(v)\leq 3\}$ has at
most $17$ constraints, hence we can extend the coloring to the whole
graph, a contradiction.
  \end{proof}

\begin{lemma}\label{lem:triangle} 
  Every 2-connected subgraph of $H(G)$ that contains exactly three
  support vertices is a cycle.
\end{lemma}

\begin{proof}
  Suppose by contradiction that $H(G)$ contains a 2-connected subgraph
  $C$ of size $\geq 3$ that has exactly three support vertices
  $S=\{s_1,s_2,s_3\}$ and that is not a cycle.

  Suppose by contradiction that $C$ contains no cycle $C'$ with
  $S\subseteq C'\subseteq C$.  As $C$ is 2-connected, by Menger's
  Theorem there exist two internally vertex-disjoint paths $Q,Q'$
  between $s_1,s_2$. Let $C''$ be the cycle $Q\cup Q'$. By assumption
  $C''$ does not contain $s_3$. So it contains just two support
  vertices, a contradiction to Lemma~\ref{cl:cycle2}.  So $C$
  contains a cycle $C'$ with $S\subseteq C'\subseteq C$.

  By Lemma~\ref{lem:cycleimpair}, cycle $C'$ contains a subpath
  $x_1v_1x_2v_2x_3$ where $x_1,x_2,x_3$ are support vertices of Type
  ($S_3$) and $v_1,v_2$ are vertices of degree 2. As $C$ contains just
  three support vertices, we have $S=\{x_1,x_2,x_3\}$. Vertices
  $x_1,x_3$ are support vertices of Type \usss, they are of degree $3$
  and only adjacent to positive vertices and to vertices of degree $2$
  so they are not adjacent.  The graph $H(G)$ contains only edges
  incident to support vertices, so there exists a vertex $y$ of $C'$
  adjacent to $x_1,x_3$, and $x_1v_1x_2v_2x_3y$ is the cycle
  $C'$. 
If $C'$ has some chords in $H(G)$, then $H(G)$ contains a cycle
  with two support vertices only, a contradiction to
  Lemma~\ref{cl:cycle2}. So $C'$ is an induced cycle of $H(G)$ and so
  $C'$ has strictly less vertices than $C$. Let $y'$ be a vertex of $C$
  distinct from $x_1,v_1,x_2,v_2,x_3,y$. Vertex $y'$ is not a support
  vertex, $C$ is 2-connected and $H(G)$ contains only edges incident
  to support vertices, so $y'$ is adjacent to at least two vertices in
  $S$. Then $H(G)$ contains a cycle with two support vertices only, a
  contradiction to Lemma~\ref{cl:cycle2}.

\end{proof}

We need the following lemma from Brooks~\cite{b41}:

\begin{lemma}[\cite{b41}]
\label{claim:Brooks} If $G$ is a $2$-connected graph
that is neither a clique nor an odd cycle, and $L$ is a list
assignment on the vertices of $G$ such that $\forall u \in
V(G), |L(u)|\geq d(u)$, then $G$ is $L$-colorable.
\end{lemma}

\begin{lemma}\label{lem:nest} 
  Every 2-connected subgraph of $H(G)$ of size at least three is
  either a cycle with an odd number of support vertices or a subgraph
  of a lock of $H(G)$.
\end{lemma}
\begin{proof}
 

  Suppose by contradiction that $H(G)$ contains a 2-connected subgraph
  $C$ of size $\geq 3$ that is not a cycle with an odd number of
  support vertices nor a subgraph of a lock of $H(G)$.  Let
  $S=\{s_1,\ldots,s_{p}\}$ be the support vertices of $C$. By
  Lemma~\ref{cl:cycle2}, $p\geq 3$. Let $\mathcal S$ be the graph with
  $V(\mathcal S)=S$ where there is an edge between $s_i$ and $s_j$ if
  and only if they are adjacent or have a common neighbor in $G$.

  \begin{claim}
    \label{cl:clique4}
   $\mathcal S$ is not a clique of size at least four.     
  \end{claim}
 
  \begin{proof}
    Suppose, by contradiction that $\mathcal S$ is a clique with $p\geq 4$.

    Given a support vertex $x$, we say that a support vertex $x'$,
    distinct from $x$,
    satisfies the property $P_x$ if it is either adjacent to $x$ in
    $G$ or has a non-positive common neighbor with $x$ in $G$.  At
    most two vertices can satisfy $P_x$ (vertex $a$ of
    Figure~\ref{fig:neg} if $x$ is of Type ($S_1$), vertices $b,c$ if
    $x$ is of Type ($S_2$), vertices $b,d$ if $x$ is of Type
    ($S_3$)). Note that if $x$ satisfies $P_x$, then $x$ satisfies
    $P_{x'}$.

    We claim that there exist two support vertices in $S$ that do not
    have a positive common neighbor in $G$. Suppose by contradiction,
    that every pair of vertices of $S$ has a positive common
    neighbor. By Lemma~\ref{lem:supportpositive}, every support vertex
    has at most one positive neighbor, so all the vertices of $S$ are
    adjacent to the same positive vertex $v$. As $C$ is 2-connected,
    there is a path $Q$ in $C\setminus \{v\}$ between $s_1,s_2$. Let
    $s_i$ be the first support vertex, distinct from $s_1$, appearing
    along $Q$ while starting from $s_1$ (maybe $i=2$ if there is no
    support vertex in the interior of $Q$). Let $Q'$ be the subpath
    of $Q$ between $s_1$ and $s_i$ (maybe $Q=Q'$). Then $Q'\cup\{v\}$
    forms a 2-connected subgraph of size $\geq 3$ with exactly two
    support vertices, a contradiction to Lemma~\ref{cl:cycle2}.  So
    there exist two support vertices $x,x'$ in $S$ that do not have a
    positive common neighbor in $G$. Since $\mathcal S$ is a clique,
    vertices $x,x'$ are adjacent or have a common non-positive
    neighbor, so $x$ satisfies $P_{x'}$ (and $x'$ satisfies $P_x$).

    Suppose there exists a support vertex $y\in S$ that does not
    satisfy $P_x$ nor $P_{x'}$. Since $\mathcal S$ is a clique, vertex
    $y$ has a common positive neighbor $z$ with $x$ and $z'$ with
    $x'$. Since $x$ and $x'$ have no positive common neighbor, $z$ and
    $z'$ are distinct. Thus $y$ has two positive neighbors, a
    contradiction. So every vertex of $S\setminus \{x,x'\}$ satisfies
    either $P_x$ or $P_{x'}$.  If two vertices $y,y'$ of $S\setminus
    \{x,x'\}$ satisfy $P_x$, then at least three vertices, $x',y,y'$
    verify $P_x$, a contradiction. So there is at most one vertex of
    $S\setminus \{x,x'\}$ satisfying $P_x$ and similarly at most one
    satisfying $P_{x'}$. So $p\leq 4$ and we can assume, w.l.o.g.,
    that $S=\{x,x',y,y'\}$, where vertex $y$ satisfies $P_x$ and not
    $P_{x'}$ and vertex $y'$ satisfies $P_{x'}$ and not $P_x$.  Thus
    $x$ has a common positive neighbor $z$ with $y'$ and $x'$ has a
    common positive neighbor $z'$ with $y$.  Since $x,x'$ do not have
    a common positive neighbor, $z$ and $z'$ are distinct. Vertices
    $y,y'$ have at most one positive neighbor, thus, they do not have
    a common positive neighbor. Since $\mathcal S$ is a clique, $y$
    satisfies $P_{y'}$. Let $(y_1,y_2,y_3,y_4)=(x,x',y',y)$ (subscript
    are understood modulo $4$).

  Suppose there exists $i\in\{1,2,3,4\}$ such that $y_i,y_{i+1}$ are
  adjacent in $G$. Two support vertices can be adjacent only if they
  are of Type ($S_1$). So $y_i,y_{i+1}$ are of Type ($S_1$) and of
  degree two. Then $y_i$ is only adjacent to $y_{i+1}$ and to a
  positive vertex in $\{z,z'\}$. If $y_{i}$ is adjacent to $y_{i-1}$,
  then $y_{i-1}=y_{i+1}$, a contradiction.  If $y_{i}$ is not adjacent
  to $y_{i-1}$, then $y_{i+1}$ is a common neighbor of $y_{i}$ and
  $y_{i-1}$. Since $y_{i+1}$ is of degree two and has a positive
  neighbor, $y_{i}=y_{i-1}$, a contradiction. So $y_i,y_{i+1}$ are not
  adjacent in $G$ for any $1\leq i\leq 4$. Let $w_i$ be a non-positive
  common neighbor of $y_i,y_{i+1}$.

  Suppose there exists $i\in\{1,2,3,4\}$ such that $d(y_i)=2$. Then
  $w_i=w_{i-1}$. So $\{y_{i-1},y_{i},y_{i+1}\}\subseteq N(w_i)$, and
  $w_i$ is not positive, so $d(w_i)= 3$.  Two support vertices can
  have a common neighbor of degree $3$ only if they are both of degree
  two (Type ($S_2$)).  So $d(y_{i-1})=d(y_{i})=d(y_{i+1})=2$.  Since
  $y_{i+1}$ is of degree two and has a positive neighbor,
  $w_{i}=w_{i+1}$, so $y_{i+2}\in N(w_i)$, a contradiction.  So
  $d(y_i)\geq 3$ for any $1\leq i\leq 4$.

  Then all the $y_i$ are of Type ($S_3$), they are of degree three and
  their non positive neighbors are of degree two. Thus $d(w_i)=2$ for
  any $1\leq i\leq 4$. So $y_1,\ldots,y_4,w_1,\ldots,w_4,z,z'$ induce
  a lock.  So all the edges incident to
  $S=\{y_1,\ldots,y_4\}=\{s_1,\ldots,s_4\}$ belong to a lock,
  contradicting the definition of $C$.
  \end{proof}

  By Lemma~\ref{cl:cycle2}, the graph $\mathcal S$ is not an edge. If
  $\mathcal S$ is a triangle, then $C$ contains exactly three support
  vertices and, by Lemma~\ref{lem:triangle}, it is a cycle with an odd
  number of support vertices, a contradiction. So $\mathcal S$ is not
  a triangle.  By Claim~\ref{cl:clique4}, $\mathcal S$ is not a clique
  of size at least $4$. So finally, $\mathcal S$ is not a clique.

  Suppose, by contradiction, that $\mathcal S$ is an odd cycle with
  $\geq 5$ vertices. Then  $C$ is a 2-connected graph that is not
  a cycle, so it contains a vertex $v$ with at least $3$ neighbors in
  $C$. If $v$ is not a support vertex, then it has at least $3$
  support neighbors in $C$ that form a triangle in $\mathcal S$, a
  contradiction. So $v$ is a support vertex. Then either $v$ has three
  neighbors in $\mathcal S$, a contradiction to $\mathcal S$ being a
  cycle, or $C$ contains a cycle with two support vertices, a
  contradiction to Lemma~\ref{cl:cycle2}. So $\mathcal S$ is not an
  odd cycle.

  Suppose, by contradiction, that $\mathcal S$ is not
  $2$-connected. Then there exist three support vertices $s,s',s''$
  of $S$ such that $s',s''$ appears in two different connected
  components of $\mathcal S\setminus\{s\}$. As $C$ is 2-connected,
  there exists a path $Q$ between $s',s''$ in $C\setminus \{s\}$. This
  path $Q$ is composed only of edges incident to support vertices so
  in $\mathcal S\setminus\{s\}$ it corresponds to a path between
  $s',s''$, a contradiction. So $\mathcal S$ is $2$-connected.

  We now consider the graph $G$, we color by minimality $G\setminus
  (S\cup\{v \in N_G(S) | d_G(v)\leq 3\})$.  We show how to color $S$.
  In the three Types ($S_j$), the number of constraints on a support
  vertex $s_i$ of Type ($S_j$) is at most $k+2$ minus the number of
  its neighbors in $\mathcal S$.  So the number of available colors of
  a support vertex is at least its degree in $\mathcal S$. Now
  Lemma~\ref{claim:Brooks} can be applied to $\mathcal S$ that is not
  a clique, not an odd cycle and 2-connected. So we can color
  $S$. Every vertex of $\{v \in N_G(S) | d_G(v)\leq 3\}$ has at most
  $17$ constraints, hence we can extend the coloring to the whole
  graph, a contradiction.
\end{proof}

A \emph{cactus} is a connected graph in which any two cycles have at
most one vertex in common.

\begin{lemma}
\label{lem:cactus}
  Every connected component of $H(G)$ is either a cactus where each
  cycle has an odd number of support vertices or a lock.
\end{lemma}

\begin{proof}
  All the edges of a lock are incident to support vertices of type
  \usss so all the edges of a lock of $G$ appear in $H(G)$.  The only
  vertices of a lock that can have neighbors outside a lock are locked
  vertices (vertices $u$ and $x$ on Figure~\ref{fig:lock}).  By
  Lemma~\ref{lem:config}, graph $G$ does not contain Configuration
  {($C_{11}$)}, so a locked vertex is incident to only two support
  vertices, the two support vertices of a lock. A lock is a connected
  component of $H(G)$.

  Let $C$ be a connected components of $H(G)$ that is not a lock. By
  Lemma~\ref{lem:nest}, each 2-connected subgraph of $C$ is a cycle
  with an odd number of support vertices. So $C$ is a cactus where each
  cycle of $C$ has an odd number of support vertices.
\end{proof}

\section{Discharging rules}\label{sect:dis} 

A \emph{negative} vertex is a support vertex of type ($S_1$) or
($S_2$) or a vertex of degree $2$ adjacent to two support vertices of
type ($S_3$). In this case we say that the negative vertex is of type ($N_1$),
($N_2$) or ($N_3$) respectively.

We design discharging rules $R_{1.1}$, $R_{1.2}$, $R_{1.3}$,
$R_{1.4}$, $R_{1.5}$, $R_2$, $R_3$, $R_4$ and $R_g$ (see
Figure~\ref{fig:rules}): for any vertex $x$ of degree at least $3$,

\begin{itemize}
\item Rule $R_1$ is when $3\leq d(x)\leq 7$, and $x$ is $1$-linked (with a
path $x-a-y$) to a vertex $y$.
\begin{itemize}
\item Rule $R_{1.1}$ is when $x$ is weak with $d(y)\leq 7$.  Then $x$
  gives $\frac{2}{5}$ to $a$.
\item Rule $R_{1.2}$ is when $x$ is not weak and $y$ is weak. Then $x$
gives $\frac{3}{5}$ to $a$
\item Rule $R_{1.3}$ is when $x$ and $y$ are not weak, with $d(y)\leq
  7$. Then $x$ gives $\frac{1}{2}$ to $a$.
\item Rule $R_{1.4}$ is when $8 \leq d(y)\leq 14$. Then $x$ gives
  $\frac{3}{8}$ to $a$.
\item Rule $R_{1.5}$ is when $15\leq d(y)$ and $a$ is not
  negative. Then $x$ gives $\frac{1}{5}$ to $a$.
\end{itemize}
\item Rule $R_2$ is when $3\leq d(x)\leq 7$ and $x$ is adjacent to a vertex
$u$ of degree $3$ that is adjacent to a vertex of degree $2$ and a
vertex of degree at most $7$. Then $x$ gives $\frac{1}{10}$ to $u$.
\item Rule $R_3$ is when $8 \leq d(x) \leq 14$. Then $x$ gives
$\frac{5}{8}$ to each of its neighbors.
\item Rule $R_4$ is when $15 \leq d(x)$. Then $x$ gives $\frac{4}{5}$
to each of its neighbors. 
\item Rule $R_g$ states that each  positive vertex gives $\frac{2}{5}$
  to a common pot, and that each negative vertex receives $\frac{1}{5}$
  from the common pot.
 
\end{itemize}

\begin{figure}
  \subfigure[$R_{1.1}$]{\makebox[2.5cm]{\input{rule11}}} \hfill
  \subfigure[\noindent $R_{1.2}$]{\makebox[2.5cm]{\input{rule12}}} \hfill
  \subfigure[$R_{1.3}$]{\makebox[2.5cm]{\input{rule14}}} \hfill
  \subfigure[$R_{1.4}$]{\makebox[2.5cm]{\input{rule15}}} \hfill
  \subfigure[$R_{1.5}$]{\makebox[2.5cm]{\input{rule16}}}
  \centering \subfigure[$R_2$]{\input{rule2}} $\qquad$
  \subfigure[$R_3$]{\makebox[3cm]{\input{rule3}}} $\qquad$
  \subfigure[$R_4$]{\makebox[3cm]{\input{rule4}}} 
  \caption{Discharging rules $R_{1.i}$, $R_2$, $R_3$, and $R_4$}
  \label{fig:rules}
\end{figure}
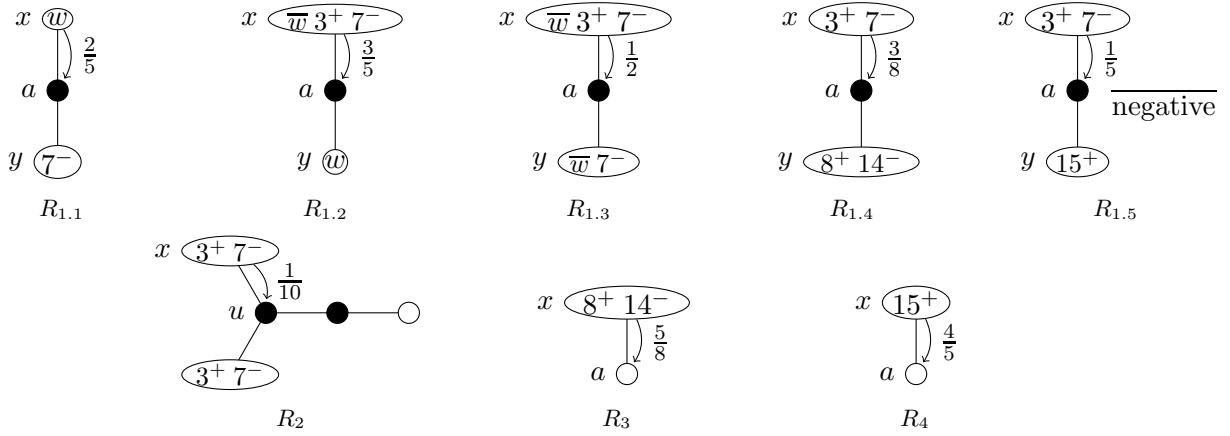

\begin{lemma}\label{lem:Rgvalid} 

The common pot has
  non-negative value after applying $R_g$.
\end{lemma}

\begin{proof}
  Given a set of vertices $X$, let $n(X)$ be its number of negative
  vertices and $p(X)$ its number of positive vertices.  To prove that
  the common pot has positive value after applying $R_g$, we show that
  each connected component $C$ of $H(G)$ satisfies $p(C)\geq
  \left\lceil\frac{n(C)}{2}\right\rceil$.

  Let $C$ be a connected components of $H(G)$.  By
  Lemma~\ref{lem:cactus}, $C$ is either a cactus where each
  cycle has an odd number of support vertices or a lock.  If $C$ is a
  lock, then $n(C)=4$ and $p(C)=2$, so we are done. So we can assume
  that $C$ is a cactus where each cycle has an odd number of support
  vertices.


  \begin{claim}
\label{cl:pendant}
Every connected subgraph $C'$ of $C$, whose pendant vertices are
positive vertices, whose support vertices are adjacent to their
positive neighbor in $C'$ and whose negative vertices of Type \nnn are
adjacent to their two neighbors in $C'$, satisfies $p(C')\geq
\left\lceil\frac{n(C')}{2}\right\rceil$.
  \end{claim}

\begin{proof}
  Suppose by contradiction that this is false. Let $C'$ be a connected
  subgraph
  of $C$ of minimum number of vertices, whose pendant vertices are
  positive vertices, whose support vertices are adjacent to their positive
  neighbor in $C'$, and such that $p(C')<
  \left\lceil\frac{n(C')}{2}\right\rceil$.  The graph $C'$ is a
  connected subgraph of a cactus so it is also a cactus.

  Suppose first that $C'$ contains a pendant vertex $u$.
  Let $x$ be the neighbor of the positive vertex $u$ in $C'$. As
  $H(G)$ contains only edges incident to support vertices, $x$ is a
  support vertex. So it is not positive and thus is not a pendant
  vertex of $C'$. So $x$ has at least two neighbors in $C'$. We
  consider different cases according to the Type of $x$ and its number
  of neighbors in $C'$.

  \begin{itemize}
  \item \emph{$x$ is of Type \us.} Then let $a$ be the neighbor of $x$
    distinct from $u$.  We have $a\in C'$ and $a$ is a support vertex
    of Type \us. The positive neighbor $b$ of $a$ is in $C'$ by
    assumption.  Let $C''$ be the graph $C'\setminus \{u,x,a\}$. We
    have $n(C'')=n(C')-2$ and $p(C'')=p(C')-1$.  The graph $C''$ is a
    connected subgraph of $C$ since $u, x, a$ is subpath of $C'$ where
    $u$ is pendant and $x,a$ are of degree $2$.  All the pendant
    vertices of $C''$ are positive since the only new possible pendant
    vertex is $b$.  All the support vertices of $C''$ are adjacent to
    their positive neighbor in $C''$ since the only positive vertex
    that has been removed is $u$ and its support neighbor $x$ has also
    been removed. All the negative vertices of Type \nnn are adjacent
    to their two neighbors in $C'$ as no support vertex of Type \usss
    has been removed. So by minimality, we have $p(C'')\geq
    \left\lceil\frac{n(C'')}{2}\right\rceil$, and so
    $p(C')=p(C'')+1\geq \left\lceil\frac{n(C'')+2}{2}\right\rceil=
    \left\lceil\frac{n(C')}{2}\right\rceil$.

  \item \emph{$x$ is of Type \uss.} Then let $a$ be the neighbor of
    $x$ distinct from $u$. We have $a\in C'$ and $a$ is of degree
    $3$. Let $b,c$ be the neighbors of $a$ distinct from $x$. Since
    $a$ is not positive, it is not a pendant vertex of $C'$, so at
    least one of $b,c$ is in $C'$. We assume w.l.o.g. that $c$ is in
    $C'$. As $H(G)$ contains only edges incident to support vertices,
    vertex $c$ is a support vertex of Type \uss. We consider two cases
    depending on whether $a$ has its three neighbors in $C'$ or
    not.

    If $b\in C'$, then let $C''$ be the graph $C'\setminus
    \{u,x\}$. We have $n(C'')=n(C')-1$ and $p(C'')=p(C')-1$.  The
    graph $C''$ is a connected subgraph of $C$, all its pendant
    vertices are positive, all its support vertices are adjacent to
    their positive neighbor in $C''$ and all its negative vertices of Type \nnn are
adjacent to their two neighbors in $C'$. So by minimality, we have $p(C'')\geq
    \left\lceil\frac{n(C'')}{2}\right\rceil$, and so $p(C')\geq
    \left\lceil\frac{n(C')}{2}\right\rceil$.

    If $b\notin C'$, then let $C''$ be the graph $C'\setminus
    \{u,x,a,c\}$.  We have $n(C'')=n(C')-2$ and $p(C'')=p(C')-1$.  The
    graph $C''$ is a connected subgraph of $C$, all its pendant
    vertices are positive, all its support vertices are adjacent to
    their positive neighbor in $C''$ and all its negative vertices of Type \nnn are
adjacent to their two neighbors in $C'$. So by minimality, we have $p(C'')\geq
    \left\lceil\frac{n(C'')}{2}\right\rceil$, and so $p(C')\geq
    \left\lceil\frac{n(C')}{2}\right\rceil$.

  \item \emph{$x$ is of Type \usss and has two neighbors in $C'$.}
    Then let $c$ be the neighbor of $x$ distinct from $u$ that is in
    $C'$.  Vertex $c$ is of degree $2$, it is not positive, so its
    neighbor $d$, distinct from $x$, is in $C'$.  As $H(G)$ contains
    only edges incident to support vertices and $c$ is not a support
    vertex, vertex $d$ is a support vertex and so of Type \usss. Let
    $e,f$ be the neighbors of $d$ distinct from $c$ where $e$ is a
    positive vertex and $f$ is a vertex of degree $2$. Vertex $e$ is
    the positive neighbor of $d$ so it is in $C'$ by assumption.  We
    consider two cases corresponding to whether $d$ has its three
    neighbors in $C'$ or not.  If $f\in C'$, then let $C''$ be the
    graph $C'\setminus \{u,x,c\}$. If $f\notin C'$, then let $C''$ be
    the graph $C'\setminus \{u,x,c,d\}$. In both cases, we have
    $n(C'')=n(C')-1$ and $p(C'')=p(C')-1$.  The graph $C''$ is a
    connected subgraph of $C$, all its pendant vertices are positive,
    all its support vertices are adjacent to their positive neighbor
    in $C''$ and all its negative vertices of Type \nnn are adjacent
    to their two neighbors in $C'$. So by minimality, we have
    $p(C'')\geq \left\lceil\frac{n(C'')}{2}\right\rceil$, and so
    $p(C')\geq \left\lceil\frac{n(C')}{2}\right\rceil$.

  \item \emph{$x$ is of Type \usss and has three neighbors in $C'$.}
    Then let $a,c$ be the neighbors of $x$ distinct from $u$. We have
    $a,c$ in $C'$. Vertex $a$ (resp. $c$) is of degree $2$, it is not
    positive, so its neighbor $b$ (resp. $d$) is in $C'$. As $H(G)$
    contains only edges incident to support vertices and $a$ and $c$
    are not support vertices, vertices $b$ and $d$ are support
    vertices and so of Type \usss. The positive neighbor $h$ of $b$
    (resp. $e$ of $d$) is in $C'$, by assumption.  We consider several
    cases corresponding to whether $b$ and $d$ have their three
    neighbors in $C'$ or not.  If $b$ and $d$ both have their three
    neighbors in $C'$, then let $C''$ be the graph $C'\setminus
    \{u,x,c,a\}$. If $b$ has its three neighbors in $C'$ but not $d$,
    then let $C''$ be the graph $C'\setminus \{u,x,c,a,d\}$.  If $d$
    has its three neighbors in $C'$ but not $b$, then let $C''$ be the
    graph $C'\setminus \{u,x,c,a,b\}$.  If none of $b$ and $d$ has its
    three neighbors in $C'$, then let $C''$ be the graph $C'\setminus
    \{u,x,c,a,b,d\}$. In the four cases we have $n(C'')=n(C')-2$ and
    $p(C'')=p(C')-1$. The graph $C''$ is not necessarily connected but
    it is composed of one or two connected subgraph of $C$ whose all
    pendant vertices are positive, all support vertices are adjacent
    to their positive neighbor in $C''$ and all its negative vertices
    of Type \nnn are adjacent to their two neighbors in $C'$. So by
    minimality (on each component of $C''$), we have $p(C'')\geq
    \left\lceil\frac{n(C'')}{2}\right\rceil$, and so $p(C')\geq
    \left\lceil\frac{n(C')}{2}\right\rceil$.
  \end{itemize}
  
  Now we can assume that $C'$ contains no pendant vertex. Suppose that
  $C'$ is a single vertex $v$. Then $v$ is not support as all support
  vertices have their positive neighbor in $C'$ and $v$ is not
  negative of Type \nnn as negative vertices of Type \nnn have their
  two neighbors in $C'$. So $v$ is not negative and $p(C')\geq
  \left\lceil\frac{n(C')}{2}\right\rceil=0$. Now we can assume that
  $C'$ is not a single vertex.  The graph $C'$ is a cactus, not a
  single vertex, contains no pendant vertex, so it contains a cycle
  $C''$, of size $\geq 3$, such that $C'''=C'\setminus C''$ is
  connected (note that we may have $C'=C''$ and $C'''$ is empty).
  Cycle $C''$ is a cycle of $C$ so it has an odd number of support
  vertices by Lemma~\ref{lem:cactus}.  Let $S$ be the set of support
  vertices of $C''$, with $s=|S|$. By Lemma~\ref{lem:cycleimpair},
  cycle $C''$ contains a subpath $s_1v_1s_2v_2s_3$ where $s_1,s_2,s_3$
  are support vertices of Type ($S_3$) and $v_1,v_2$ are vertices of
  degree 2. By assumption, the positive vertex $z$ that is adjcent to
  $s_2$ is in $C'$. It is not in $C''$ as there is no chord in $C''$.
  So the only vertex of $C''$ that has some neighbors in $C'\setminus
  C''$ is $s_2$. So all the positive vertices that are adjacent to
  $S\setminus \{s_2\}$ are vertices of $C'$ and thus of $C''$. A
  positive vertex of $C''$ has at most two support neighbors in $C''$
  so $p(C'')\geq \left\lceil\frac{s-1}{2}\right\rceil$. A support
  vertex of Type \us or \uss is a negative vertex of Type \n or \nn.
  A negative vertex of Type \nnn of $C''$ is of degree $2$ and so has
  its two neighbors on $C''$ and this two neighbors are support
  vertices of Type \usss.  So the number of negative vertices of $C''$
  is less or equal to the number of support vertices of $C''$ and
  strictly less if $C''$ contains a vertex of Type \nnn. Vertex $v_1$
  is of Type \nnn, so $s> n(C'')$ and so $p(C'')\geq
  \left\lceil\frac{s-1}{2}\right\rceil \geq
  \left\lceil\frac{n(C'')}{2}\right\rceil$.  The graph $C'''$ is a
  connected subgraph of $C$ whose all pendant vertices are positive,
  all support vertices are adjacent to their positive neighbor in
  $C'''$ and all its negative vertices of Type \nnn are adjacent to
  their two neighbors in $C'$. So by minimality we have $p(C''')\geq
  \left\lceil\frac{n(C''')}{2}\right\rceil$.  So finally,
  $p(C')=p(C'')+p(C''')\geq \left\lceil\frac{n(C''')}{2}\right\rceil+
  \left\lceil\frac{n(C'')}{2}\right\rceil\geq
  \left\lceil\frac{n(C'')+n(C''')}{2}\right\rceil=
  \left\lceil\frac{n(C')}{2}\right\rceil$.

  \end{proof}

  Let $C'$ be the graph obtained from $C$ by removing all pendant
  vertices that are not positive vertices.  We claim that $C'$ is a
  connected subgraph of $C$, whose pendant vertices are positive
  vertices, whose support vertices have their positive neighbor in
  $C'$, whose negative vertices of Type \nnn are adjacent to their two
  neighbors in $C'$ and such that $n(C')=n(C)$.  As $C$ is connected
  and only pendant vertices have been removed from $C$, the graph $C'$
  is also connected.  All support and negative vertices are of degree
  $2$ or $3$ and have all their incident edges in $H(G)$ and so in
  $C$, so there is no pendant vertex of $C$ that is a support or a
  negative vertex. So no support or negative vertex have been removed
  from $C$ and $n(C')=n(C)$. A pendant vertex of $C$ that has been
  removed is not positive, not support, not negative but incident to a
  support, so it is necessarily a degree $2$ vertex $a$ incident to a
  support vertex $x$ of Type \usss (with notations of
  Figure~\ref{fig:neg}). When $a$ is removed from $C$, this does not
  create any new pendant vertex as $x$ has degree $2$ after the
  removal. All pendant vertices that are not positive are removed from
  $C$, no new pendant vertices are created, thus in $C$ all pendant
  vertices are positive.  No positive vertex has been removed and
  each support vertex is adjacent to its positive neighbor in $H(G)$,
  so support vertices of $C'$ are adjacent to their positive neighbor
  in $C'$. No support vertex have been removed and each negative
  vertex of Type \nnn is adjacent to its support neighbors of Type
  \usss in $H(G)$, so negative vertices of Type \nnn of $C'$ are
  adjacent to their two neighbors in $C'$.  By Claim~\ref{cl:pendant}
  applied to $C'$, we have $p(C')\geq
  \left\lceil\frac{n(C')}{2}\right\rceil$. So $p(C)= p(C')\geq
  \left\lceil\frac{n(C')}{2}\right\rceil=
  \left\lceil\frac{n(C)}{2}\right\rceil$ and we are done.
\end{proof}


We now use the discharging rules to prove the following:

\

\begin{lemma}\label{lem:mad} 
$\mad(G) \geq 3$.
\end{lemma}

\begin{proof}
  We attribute to each vertex a weight equal to its degree, and apply
  discharging rules $R_1$, $R_2$, $R_3$, $R_4$ and $R_g$. The common
  pot is empty at the beginning and, by Lemma~\ref{lem:Rgvalid}, it
  has non-negative value after applying $R_g$. We show that all the
  vertices have a weight of at least $3$ at the end.

  Let $u$ be a vertex of $G$.  By Lemma~\ref{lem:config}, graph $G$
  does not contain Configurations {($C_1$)} to {($C_{11}$)}.
  According to Configuration {($C_1$)}, we have $d(u) \geq 2$. We now
  consider different cases corresponding to the value of $d(u)$. 

\begin{enumerate}

\item $d(u)=2$.

  So $u$ has an initial weight of $2$ and gives nothing. We show that it
  receives at least $1$, so it has a final weight of at least $3$.

\begin{enumerate}

\item \textit{Assume $u$ is adjacent to a vertex $u_2$ of degree
    $2$}.\\ Then $u$ is a negative vertex of Type \n and receives
  $\frac{1}{5}$ from the common pot by $R_g$.  According to
  Configuration {($C_2$)}, vertex $u$ is adjacent to a vertex $v$ with
  $d(v)=k$. Since $k \geq 17$, according to $R_4$, vertex $v$ gives
  $\frac{4}{5}$ to $u$. 

\item \textit{Assume both neighbors $v_1$ and $v_2$ of $u$ are of
degree at least $3$}. \\
Vertex $u$ is not a negative vertex of Type \n since it has no
neighbor of degree $2$.

\begin{enumerate}
\item \textit{$u$ has two weak neighbors}\\
  Then $u$ is a negative vertex of Type \nnn. It receives
  $\frac{1}{5}$ from the common pot by $R_g$ and $\frac{2}{5}$ from
  each of its two neigbors by $R_{1.1}$.
\item \textit{$u$ has one weak neighbor $w$ and one non-weak neighbor $v$}
  \begin{enumerate}
  \item \textit{$3\leq d(v)\leq 7$}\\
    Vertex $u$ receives $\frac{3}{5}$ from $v$ by $R_{1.2}$ and
    $\frac{2}{5}$ from $w$ by $R_{1.1}$.
  \item \textit{$8\leq d(v)\leq 14$}\\
    Vertex $u$ receives $\frac{5}{8}$ from $v$ by $R_{3}$ and
    $\frac{3}{8}$ from $w$ by $R_{1.4}$.
  \item \textit{$15\leq d(v)$ }\\
    Vertex $w$ is weak and $v$ has degree at least $15$, so one can
    check that $u$ is not negative of Type \n or \nnn. According to
    Configuration {($C_3$)}, it not negative of Type \nn. So $u$ is
    not negative and it receives $\frac{1}{5}$ from $w$ by $R_{1.5}$
    and $\frac{4}{5}$ from $v$ by $R_{4}$.
  \end{enumerate}

\item \textit{$u$ has two non-weak neighbors $v,v'$}
  \begin{enumerate}
  \item \textit{$3\leq d(v)\leq 7$ and $3\leq d(v')\leq 7$ }\\
    Vertex $u$ receives $\frac{1}{2}$ from each neighbor by $R_{1.3}$.
  \item \textit{$3\leq d(v)\leq 7$ and $8\leq d(v')\leq 14$ }\\
    Vertex $u$ receives $\frac{5}{8}$ from $v'$ by $R_{3}$ and
    $\frac{3}{8}$ from $v$ by $R_{1.4}$.
  \item \textit{$3\leq d(v)\leq 7$ and $15\leq d(v')$ }\\
    If $u$ is negative, it receives $\frac{1}{5}$ from the common pot
    by $R_g$.  If $u$ is non-negative, it receives $\frac{1}{5}$ from
    $v$ by $R_{1.5}$. In both cases, it receives $\frac{4}{5}$ from
    $v'$ by $R_{4}$.
  \item \textit{$8\leq d(v)$ and $8\leq d(v')$ }\\
    Vertex $u$ receives at least $\frac{5}{8}$ from each neighbor by
    $R_{3}$ or $R_{4}$ .
  \end{enumerate}

\end{enumerate}

\end{enumerate}
\item $d(u)=3$.\\  So $u$ has an initial weight of $3$. We show that it
  has a final weight of at least $3$.

\begin{enumerate}

\item \textit{Assume $u$ has three neighbors $y_1$, $y_2$ and $y_3$ of
    degree $2$.}\\ Let $z_i$, $1\leq i\leq 3$, be the neighbors of
  $y_i$ distinct from $u$. According to Configuration {($C_3$)},
  $d(z_1)=d(z_2)=d(z_3)=k$. So $y_1$, $y_2$ and $y_3$ are negative
  vertices of Type \nn. So no rule applies to $u$.

\item \textit{Assume $u$ has exactly two neighbors $y_1$ and $y_2$ of
    degree $2$.}\\ Let $z_i$, $1\leq i\leq 2$, be the neighbors of
  $y_i$ distinct from $u$.  Let $x$ the third neighbor of $u$,
  $d(x)\geq 3$.  
According to Configuration
  {($C_3$)}, we are in one of the two following cases:

\begin{enumerate}

\item \textit{$d(x) \geq k-2$.}\\
  Vertex $x$ gives $\frac{4}{5}$ to $u$ by $R_4$ and $u$ gives nothing
  to $x$.

\begin{enumerate}
\item \textit{Assume vertex $u$ is weak.}\\
  Since $u$ is weak, $d(y_i)\leq 14$, so vertex $u$ gives at most
  $\frac{2}{5}$ to each of $y_1,y_2$ by $R_{1.1}$ or $R_{1.4}$.

\item \textit{Assume vertex $u$ is not weak.}\\
  Then, w.l.o.g., $d(z_1)\geq 15$. So vertex $u$ gives at most
  $\frac{1}{5}$ to $y_1$ by $R_{1.5}$. Vertex $u$ gives at most
  $\frac{3}{5}$ to $y_2$ by $R_{1.2}$, $R_{1.3}$, $R_{1.4}$ or
  $R_{1.5}$.
\end{enumerate}

\item \textit{$d(z_1)=d(z_2)=k$.}

\begin{enumerate}

\item \textit{$d(x) \leq 7$.}\\
  According to Configuration {($C_4$)}, vertex $u$ gives nothing to
  $x$ by $R_2$.  Vertices $y_1$ and $y_2$ are negative (of Type \nn)
  and $u$ gives nothing to $y_1$, $y_2$.

\item \textit{$d(x) \geq 8$.}\\
  Vertex $u$ gives $\frac{1}{5}$ to $y_1$ and $y_2$ by $R_{1.5}$.
  Vertex $x$ gives at least $\frac{5}{8}$ to $u$ by $R_3$ or $R_4$.
\end{enumerate}
\end{enumerate}

\item \textit{Assume $u$ has exactly one neighbor $y$ of degree $2$}\\
  Let $z$ be the neighbor of $y$ distinct from $u$.  Let $w$ and $x$
  the other neighbors of $u$, where $d(w)\geq d(x)\geq 3$. We consider
  three cases according to the value of $d(w)$.
\begin{enumerate}

\item \textit{$15 \leq d(w)$.}\\ Then, vertex $u$ gives at most
  $\frac{3}{5}$ to $y$ by $R_{1.i}$, $1\leq i \leq 5$.  Vertex $u$
  gives at most $\frac{1}{10}$ to $x$ by $R_2$.  Vertex $w$ gives
  $\frac{4}{5}$ to $u$ by $R_4$.

\item \textit{$8 \leq d(w) \leq 14$.}\\
  According to Configuration {($C_4$)}, vertex $u$ gives nothing to
  $x$ by $R_2$. Vertex $u$ gives at most $\frac{3}{5}$ to $y$ by
  $R_{1.i}$, $1\leq i \leq 5$.  Vertex $w$ gives $\frac{5}{8}$ to $u$
  by $R_3$.

\item \textit{$d(w) \leq 7$.}\\
  According to Configuration {($C_4$)}, vertex $u$ gives nothing to
  $x$ and $w$ by $R_2$. According to Configuration {($C_3$)}, we have
  $d(z)=k$.
Vertex $u$ gives $\frac{1}{5}$ to $y$ by
  $R_{1.5}$.  Both $w$ and $x$ give $\frac{1}{10}$ to $u$ by
  $R_{2}$.  
\end{enumerate}

\item \textit{Assume all the neighbors of $u$ have degree at least
    $3$ and at most $7$}.\\
According to Configuration {($C_4$)}, vertex $u$ gives nothing to
its neighbors by $R_2$.

\item \textit{Assume $u$ has no neighbor of degree $2$ and at least a neighbor
    $v$
    of degree at least $8$}.\\
  Vertex $v$ gives at least $\frac{5}{8}$ to $u$ by $R_3$ or $R_4$.
  Vertex $u$ gives at most $\frac{1}{10}$ to each of its other
  neighbors by $R_2$.
\end{enumerate}

\item $d(u)=4$.\\
So $u$ has an initial weight of $4$. We show that it
  has a final weight of at least $3$.

\begin{enumerate}

\item \textit{Assume $u$ has at least three neighbors $y_1$, $y_2$ and
    $y_3$ of degree $2$}\\ Let $z_i$ be the neighbors of $y_i$
  distinct from $u$. We assume that $d(z_1)\geq d(z_2)\geq
  d(z_3)$. Let $x$ be the neighbor of $u$ distinct from $y_1$, $y_2$
  and $y_3$. We consider three cases depending on $d(z_2)$ and
  $d(z_3)$.
\begin{enumerate}

\item $d(z_2)\leq 14$.\\
  According to Configuration {($C_7$)}, we have $d(x) \geq
  k-2$. Vertex $u$ gives at most $3 \times \frac{3}{5}$ by $R_{1.i}$,
  $1\leq i\leq 5$. Vertex $x$ gives $\frac{4}{5}$ to $u$ by $R_4$.

\item $d(z_2)\geq 15$ and $d(z_3)\leq 14$.\\
  According to Configuration {($C_6$)}, we have $d(x) \geq 8$.  Vertex
  $u$ gives at most $ \frac{1}{5}$ to each of $y_1,y_2$ by
  $R_{1.5}$. Vertex $u$ gives at most $\frac{3}{5}$ to $y_3$ by
  $R_{1.i}$.

\item $d(z_3) \geq 15$.\\Vertex $u$ gives at most $\frac{1}{5}$
  to each of its neighbors by $R_{1.5}$.
\end{enumerate}

\item \textit{Assume $u$ has exactly two neighbors $y_1$ and $y_2$ of
    degree $2$}\\ Let $z_i$ be the neighbors of $y_i$ distinct from
  $u$.  We assume that $d(z_1)\geq d(z_2)$. Let $w$ and $x$ the
  neighbors of $u$ distinct from $y_1$, $y_2$. We assume that $d(w)
  \geq d(x)\geq 3$. We consider two cases depending on $d(z_1)$.
\begin{enumerate}

\item $d(z_1) \leq 14$.\\
  According to Configuration {($C_7$)}, we have $d(w) \geq 9$.  Vertex
  $u$ gives at most $\frac{3}{5}$ to each of $y_1,y_2$ by
  $R_{1.i}$, and at most $\frac{1}{10}$ to $x$ by $R_2$.
  Vertex $x$ gives at least $\frac{5}{8}$ to $u$ by $R_3$ or $R_4$.

\item $d(z_1) \geq 15$.\\
  Vertex $u$ gives at most $\frac{1}{5}$ to $y_1$ by $R_{1.6}$, at
  most $\frac{3}{5}$ to $y_2$ by $R_{1.i}$, and at most $\frac{1}{10}$
  to each of $w,x$ by $R_2$.
\end{enumerate}

\item \textit{Assume $u$ has at most one neighbor of degree
    $2$}.\\
  Vertex $u$ gives at most $3 \times \frac{1}{10}$ by $R_{2}$, and at
  most $\frac{3}{5}$ by $R_{1.i}$.
\end{enumerate}

\item $d(u)=5$.\\
  So $u$ has an initial weight of $5$. We show that it has a final
  weight of at least $3$.

\begin{enumerate}

\item \textit{Assume $u$ has at least four neighbors $y_1$, $y_2$,
    $y_3$ and $y_4$ of degree $2$}\\
  Let $z_i$ be the neighbors of $y_i$ distinct from $u$.  We assume
  that $d(z_1)\geq d(z_2)\geq d(z_3) \geq d(z_4)$. Let $x$ be the
  neighbor of $u$ distinct from the $y_i$'s.  We consider two cases
  depending on $d(z_4)$.

\begin{enumerate}

\item $d(z_4) \leq 7$.\\
  According to Configuration {($C_8$)}, we have $d(x) \geq 8$.  Vertex
  $u$ gives at most $\frac{3}{5}$ to each of $y_i$ by $R_{1.i}$.
  Vertex $x$ gives at least $\frac{5}{8}$ to $u$ by $R_3$ or $R_4$.

\item $d(z_4) \geq 8$.\\ Vertex $u$ gives at most $5 \times
  \frac{3}{8}$ to each of $y_i$ and $x$ by $R_{1.4}$ or $R_{1.5}$.
\end{enumerate}

\item \textit{Assume $u$ has at most three neighbors of degree
    $2$}.\\
  Vertex $u$ gives at most $3 \times \frac{3}{5}$ by $R_{1.i}$, and at
  most $2\times \frac{1}{10}$ by $R_{2}$.

\end{enumerate}

\item $d(u)=6$.\\
  So $u$ has an initial weight of $6$. We show that it has a final
  weight of at least $3$.

\begin{enumerate}

\item \textit{Assume $u$ has at least five neighbors $y_1, \ldots,
    y_5$,
    of degree $2$}\\
  Let $z_i$ be the neighbors of $y_i$ distinct from $u$.  We assume
  that $d(z_1)\geq \cdots \geq d(z_5)$. Let $x$ be the neighbors of
  $u$ distinct from $y_i$'s.  According to Configuration {($C_9$)}, we
  are in one of the following two cases.

\begin{enumerate}

\item $d(z_5) \geq 8$.\\
  Vertex $u$ gives at most $6\times \frac{3}{8}$ to each of its
  neighbors by $R_{1.4}$ or $R_{1.5}$.

\item $d(x) \geq 8$.\\
  Vertex $u$ gives at most $5 \times \frac{3}{5}$ to each of $y_i$.
\end{enumerate}

\item \textit{Assume $u$ has at most four neighbors of degree
    $2$}.\\
  Vertex $u$ gives at most $4 \times \frac{3}{5}$ by $R_{1.i}$, and at
  most $2\times \frac{1}{10}$ by $R_{2}$.

\end{enumerate}

\item $d(u)=7$.\\
 So $u$ has an initial weight of $7$. We show that it has a final
  weight of at least $3$.

\begin{enumerate}

\item \textit{Assume $u$ has at least six neighbors of degree $2$
    adjacent to vertices of degree at most $3$}.\\
  According to Configuration {($C_{10}$)}, vertex $u$ has a neighbor
  $v$ of degree at least $8$.  Vertex $u$ gives at most $6 \times
  \frac{3}{5}$ by $R_{1.i}$.

\item \textit{Assume $u$ has at most five neighbors of degree $2$
    adjacent to vertices of degree at most $3$}.\\Vertex $u$ gives at
  most $5 \times \frac{3}{5}$ by $R_{1.i}$, and at most $2 \times
  \frac{1}{2}$.
\end{enumerate}

\item $8 \leq d(u) \leq 14$.\\
Then Rule $R_3$  applies to every
neighbor of $u$, and $d(u)-(d(u) \times \frac{5}{8})\geq 3$.

\item $15 \leq d(u) < k$.\\Then Rule $R_4$ applies to every neighbor
  of $u$, and $d(u)-( d(u) \times \frac{4}{5}) \geq 3$.

\item $d(u)=k$.\\Then Rule $R_4$ applies to every neighbor of $u$
 and $R_g$ applies to $u$. We have $k \geq 17$ so $k -( k \times
\frac{4}{5}+ \frac{2}{5})\geq 3$.
\end{enumerate}

Consequently, after application of the discharging rules, every vertex
$v$ of $G$ has a weight of at least $3$, meaning that $\sum_{v \in G}
d(v) \geq \sum_{v \in G} 3 = 3 |V|$. Therefore, $\mad(G) \geq 3$.
\end{proof}

Finally, $k$ is a constant integer greater than $17$ and $G$ is a
minimal graph such that $\Delta(G) \leq k$ and $G$ admits no
$2$-distance $(k+2)$-list-coloring.  By Lemma~\ref{lem:mad}, we have
$\mad(G) \geq 3$. So Theorem~\ref{thm:main} is true.

\section{Conclusion}


We proved that graphs with $\Delta(G) \geq 17$ and maximum average degree less than $3$ are list $2$-distance $(\Delta(G)+2)$-colorable. The key idea in the proof is to use Brooks' lemma (Lemma~\ref{claim:Brooks}) instead of the usual special case of an even cycle being 2-choosable. Thus we can prove stronger structural properties, which results in a global arborescent structure that is a cactus. As far as we know, Brooks' lemma has not been used in a global discharging proof before, and it might be useful for other problems. One remaining question would be to determine the maximum $\Delta(G)$ of a graph $G$ with $\mad(G)<3$ that is not $2$-distance $(\D+2)$-colorable. By Theorem~\ref{thm:main}, it cannot be more than $16$.

Note that these proofs can be effortlessly transposed to list injective $(\Delta(G)+1)$-coloring. Indeed, every vertex we color has a neighbor that is already colored. This means that in the case of list injective coloring, every vertex we color has at least one constraint less than in the case of list $2$-distance coloring. Consequently, $\Delta(G)+1$ colors are enough in the case of list injective coloring, as mentioned in the introduction.

In contrast to Theorem~\ref{thm:main}, other results have been obtained on the 2-distance coloring of planar graphs of girth at least $6$ when more colors are allowed. For example, Bu and Zhu~\cite{bz11} proved that every planar graph $G$ of girth at least $6$ was $2$-distance $(\D+5)$-colorable.

\bibliographystyle{plain}

\end{document}

%% file: weak.tex
\begin{tikzpicture}[scale=0.95]    \configtikz

\draw (0,0) node[whitenode] (u) {} -- ++(90:1cm) node[blacknode] (x)
[label=right:$x$] {} -- ++(45:1cm) node[blacknode] (u1) {} --
++(45:1cm) node[tnode] (u2) {$14^-$};

\draw (x) -- ++(135:1cm) node[blacknode] (v1) {} -- ++(135:1cm)
node[tnode] (v2) {$14^-$};
 
\end{tikzpicture}

%% file: neg1.tex
\begin{tikzpicture}[scale=0.95]   \configtikz
  \draw (0,0) node[whitenode] (u) [label=-90:$u$] {}
  -- ++(0:1cm) node[blacknode] (x) [label=-90:$x$] {}
  -- ++(0:1cm) node[blacknode] (w1) [label=-90:$a$] {}
  -- ++(0:1cm) node[whitenode] (w) [label=-90:$b$] {};
\end{tikzpicture}

%% file: neg2.tex
\begin{tikzpicture}[scale=0.95]   \configtikz
  \draw node[whitenode] (u) [label=-90:$u$] {}
  -- ++(0:1cm) node[blacknode] (x) [label=-90:$x$] {}
  -- ++(0:1cm) node[blacknode] (w1) [label=-90:$a$] {}
  -- ++(0:1cm) node[blacknode] (w2) [label=-90:$c$] {}
  -- ++(0:1cm) node[whitenode] (w) [label=-90:$d$] {};
  \draw (w1)
  -- ++(90:1cm) node[tnode] (z) [label=90:$b$] {$7^-$};
\end{tikzpicture}

%% file: neg3.tex
\begin{tikzpicture}[scale=0.95]   \configtikz  
  \draw (5,-2) node[whitenode] (u) [label=-90:$u$] {}
  -- ++(0:1cm) node[blacknode] (x) [label=-90:$x$] {}
  -- ++(0:1cm) node[blacknode] (w1) [label=-90:$c$] {}
  -- ++(0:1cm) node[blacknode] (w2) [label=-90:$d$] {}
  -- ++(0:1cm) node[whitenode] (w) [label=-90:$e$] {};
  
  \draw (x)
  -- ++(90:1cm) node[blacknode] (x1) [label=-180:$a$] {}
  -- ++(90:1cm) node[tnode] (x2) [label=-180:$b$] {$14^-$};
  
  \draw (w2)
  -- ++(90:1cm) node[blacknode] (x1) [label=-0:$f$] {}
  -- ++(90:1cm) node[tnode] (x2) [label=-0:$g$] {$14^-$};
\end{tikzpicture}

%% file: lock.tex
\begin{tikzpicture}[scale=0.95] \configtikz
 
\draw (0,0) node[whitenode] (x) [label=-180:$u$] {} --
++(45:1.4121356cm) node[blacknode] (v1)[label=90:$v_1$] {} -- ++(0:1cm)
node[blacknode] (a) {} -- ++(0:1cm) node[blacknode] (w1)[label=90:$w_1$] {} --
++(-45:1.4121356cm) node[whitenode] [label=0:$x$] (w) {};

\draw (x) -- ++(-45:1.4121356cm) node[blacknode] (v2)[label=-90:$v_2$] {} -- ++(0:1cm)
node[blacknode] (d) {} -- ++(0:1cm) node[blacknode] (w2)[label=-90:$w_2$] {};
 
\draw (v1) -- ++(-45:2cm) node[blacknode] (c) {};
 
\draw (v2) -- ++(45:2cm) node[blacknode] (b) {};
 
\draw (w2) edge node {} (w); \draw (c) edge node {} (w2); \draw (b)
edge node {} (w1);
\end{tikzpicture}

%% file: c01.tex
\begin{tikzpicture}[scale=0.95] \configtikz
   \draw node[whitenode] (u) [label=-0:$u$] {\small{$1^-$}};
\end{tikzpicture}

%% file: c02.tex
\begin{tikzpicture}[scale=0.95] \configtikz
\draw node[tnode] (w1) [label=90:$x$]
{\small{$(k-1)^-$}} -- ++(-90:1cm) node[blacknode] (u1)
[label=right:$v$] {} -- ++(-90:1cm) node[blacknode] (u2)
[label=right:$u$] {} -- ++(-90:1cm) node[whitenode] (w2)
[label=right:$w$] {};
\end{tikzpicture}

%% file: c03.tex
\begin{tikzpicture}[scale=0.95] \configtikz
\node[texte, right=1pt] at (-3,0) {\footnotesize $d(x)+d(w)$};
\node[texte, right=1pt] at (-2.8,-0.5) {\footnotesize $\leq k-1$};

\draw node[tnode] (w1) 
{\small{$(k-1)^-$}} -- ++(150:1cm) node[blacknode] (u1)
[label=right:$v$] {} -- ++(150:1cm) node[blacknode] (u2)
[label=right:$u$] {} -- ++(90:1cm) node[whitenode] (w2)
[label=90:$x$] {};
\draw (u2) -- ++(-150:1cm) node[whitenode] (y3) [label=90:$w$] {};
\end{tikzpicture}

%% file: c04.tex
\begin{tikzpicture}[scale=0.95] \configtikz
\draw node[tnode] (w1) 
{} -- ++(90:1cm) node[blacknode] (w3) [label=right:$y$] {} -- ++(90:1cm) node[blacknode]
(u1) [label=right:$v$] {} -- ++(150:1cm) node[blacknode] (u2)
[label=right:$u$] {} -- ++(90:1cm) node[tnode] (w2) [label=90:$x$]
{};

\draw (u2) -- ++(-150:1cm) node[whitenode] (y3) [label=90:$w$] {};

\draw (u1) -- ++(30:1cm) node[whitenode] (y3) [label=right:$z$] {\small{$7^-$}};

\node[texte, right=1pt] at (-2.1,1.5) {\footnotesize $d(x)+d(w)$};
\node[texte, right=1pt] at (-1.9,1.0) {\footnotesize $\leq k-1$};

\end{tikzpicture}

%% file: c05.tex
\begin{tikzpicture}[scale=0.95] \configtikz
\draw node[tnode] (w1) 
{\small{$14^-$}} -- ++(150:1cm) node[blacknode] (u1) [label=90:$v$]
{} -- ++(150:1cm) node[blacknode] (u2) [label=70:$u$] {} -- ++(90:1cm)
node[tnode] (w2) [label=90:$x$] {\small{$(k-1)^-$}};

\draw (u2) -- ++(-150:1cm) node[blacknode] (b2) [label=90:$w$] {} --
++(-150:1cm) node[tnode] (y3) {\small{$14^-$}}; 
\end{tikzpicture}

%% file: c06.tex
\begin{tikzpicture}[scale=0.95] \configtikz
\draw node[tnode] (v) [label=90:$x$] {\small{$3^-$}} -- ++(-90:1cm)
node[blacknode] (x) [label=-15:$u$] {} -- ++(-90:1cm) node[blacknode]
(y) [label=right:$v$] {} -- ++(-90:1cm) node[tnode] (u) {\small{$14^-$}};

\draw (x) -- ++(0:1cm) node[whitenode] (y3b) [label=right:$y$] {\small{$3^-$}};
 
\draw (x) -- ++(180:1cm) node[whitenode] (y4b) [label=left:$w$] {\small{$7^-$}};

\end{tikzpicture}

%% file: c07.tex
\begin{tikzpicture}[scale=0.95] \configtikz
\draw  node[whitenode] (v) [label=90:$x$] {} --
++(-90:1cm) node[blacknode] (x) [label=-15:$u$] {} -- ++(-90:1cm)
node[blacknode] (y) [label=right:$v$] {} -- ++(-90:1cm) node[tnode]
(u) {\small{$14^-$}}; 

\draw (x) -- ++(0:1cm) node[whitenode] (y3b) [label=90:$y$] {};

\draw (x) -- ++(180:1cm) node[blacknode] (y4bc) [label=-90:$w$] {} --
++(180:1cm) node[tnode] (y4b) {\small{$14^-$}};

\node[texte, right=1pt] at (-2.2,-2) {\footnotesize $d(x)+d(y)$};
\node[texte, right=1pt] at (-2,-2.5) {\footnotesize $\leq k-1$};
\end{tikzpicture}

%% file: c08.tex
\begin{tikzpicture}[scale=0.95] \configtikz
\draw node[tnode] (v) {\small{$7^-$}} -- ++(-130:1cm) 
node[blacknode] (z) [label=90:$v$] {} -- ++(-130:1cm) 
node[blacknode] (x) [label=right:$u$] {};

\draw (x) -- ++(130:1cm) node[blacknode] (y3) [label=90:$z$] {} -- ++(130:1cm)
node[whitenode] (y3b) {};

\draw (x) -- ++(-30:1cm) node[whitenode] (y4b) [label=0:$w$] {\small{$7^-$}};

\draw (x) -- ++(-90:1cm) node[whitenode] (y2b) [label=-90:$x$] {\small{$3^-$}};

\draw (x) -- ++(-150:1cm) node[whitenode] (y4b) [label=-180:$y$] {\small{$3^-$}};

\end{tikzpicture}

%% file: c09.tex
\begin{tikzpicture}[scale=0.95] \configtikz
\draw node[tnode] (v) {\small{$7^-$}} -- ++(-90:1cm) 
node[blacknode] (z) [label=right:$v$] {} -- ++(-90:1cm) 
node[blacknode] (x) [label=right:$u$] {} -- ++(30:1cm) 
node[blacknode] (y) [label=80:$t$] {} -- ++(30:1cm) node[tnode] (u) {};

\draw (x) -- ++(150:1cm) node[blacknode] (y3) [label=90:$z$] {} -- ++(150:1cm)
node[whitenode] (y3b) {};

\draw (x) -- ++(-30:1cm) node[whitenode] (y4b) [label=0:$w$] {\small{$7^-$}};

\draw (x) -- ++(-90:1cm) node[whitenode] (y2b) [label=-90:$x$] {\small{$3^-$}};

\draw (x) -- ++(-150:1cm) node[whitenode] (y4b) [label=-180:$y$] {\small{$3^-$}};

\end{tikzpicture}

%% file: c10.tex
\begin{tikzpicture}[scale=0.95] \configtikz
\draw node[tnode] (v) [label=90:$v$] {\small{$7^-$}} -- ++(-90:1cm)
node[blacknode] (z) [label=-90:$u$] {} -- ++(38.58:1cm)
node[blacknode] (x) [label=90:$w_1$] {} -- ++(38.58:1cm)
node[tnode] (u) {\small{$3^-$}};

\draw (z) -- ++(-12.86:1cm) node[blacknode] (x2) [label=5:$w_2$]
{} -- ++(-12.86:1cm) node[tnode] (u2) {\small{$3^-$}};

\draw (z) -- ++(-64.29:1cm) node[blacknode] (x3) [label=right:$w_3$]
{} -- ++(-64.29:1cm) node[tnode] (u3) 
{\small{$3^-$}};

\draw (z) -- ++(-115.71:1cm) node[blacknode] (x4) [label=left:$w_4$]
{} -- ++(-115.71:1cm) node[tnode] (u4) {\small{$3^-$}};

\draw (z) -- ++(-167.14:1cm) node[blacknode] (x5) [label=175:$w_5$]
{} -- ++(-167.14:1cm) node[tnode] (u5) {\small{$3^-$}};

\draw (z) -- ++(-218.57:1cm) node[blacknode] (x6) [label=90:$w_6$]
{} -- ++(-218.57:1cm) node[tnode] (u6) {\small{$3^-$}};

\end{tikzpicture}

%% file: c11.tex
\begin{tikzpicture}[scale=0.95] \configtikz

\draw node[whitenode] (v3) [label=90:$x$] {} -- ++(0:1cm)
node[tnode] (u) [label=90:$u$] {\small{$k$}} -- ++(45:1.41421356cm)
node[blacknode] (v1) [label=90:$v$] {} -- ++(0:1cm) 
node[blacknode] (v1b)[label=90:$z_1$] {} -- ++(0:1cm) 
node[tnode] (v1c) [label=90:$y_1$] {\small{$14^-$}}; 
 
\draw (u) -- ++(-45:1.41421356cm) 
node[blacknode] (v2) [label=-90:$w$] {} -- ++(0:1cm)
node[blacknode] (v2b) [label=-90:$z_4$] {} -- ++(0:1cm) 
node[tnode] (v2c) [label=-90:$y_2$] {\small{$14^-$}};
 
\draw (v1) -- ++(-26.56505:1.11803cm) 
node[blacknode] (v1d) [label=0:$z_2$] {} -- ++(-26.56505:1.11803cm) 
node[blacknode] (w) [label=90:$y$] {} -- ++(0:1cm) node[whitenode] (w1) {};
  
\draw (w) -- ++(-153.4349:1.11803cm) node[blacknode] (v2d) [label=0:$z_3$] {};
  
\draw (v2) edge node {} (v2d);
  
\node[texte, right=1pt] at (-0.5,-0.50) {\footnotesize support};

\end{tikzpicture}

%% file: rule11.tex
\begin{tikzpicture}[scale=0.95] \configtikz
\draw node[tnode] (x2) [label=left:$x$] {$w$} -- ++(-90:1cm)
node[blacknode] (a2) [label=left:$a$] {} -- ++(-90:1cm) 
node[tnode] (y2) [label=left:$y$] {$7^-$}; 
\draw (x2) edge [post,bend left] node [label=right:$\frac{2}{5}$] {}
(a2); 
\end{tikzpicture}

%% file: rule12.tex
\begin{tikzpicture}[scale=0.95] \configtikz
\draw (2,0) node[tnode] (x2) [label=left:$x$] {\small{$\neg w$ $3^+$ $7^-$}} --
++(-90:1cm) node[blacknode] (a2) [label=left:$a$] {} -- ++(-90:1cm)
node[whitenode] (y2) [label=left:$y$] {$w$}; \draw (x2) edge
[post,bend left] node [label=right:$\frac{3}{5}$] {} (a2);

\end{tikzpicture}

%% file: rule14.tex
\begin{tikzpicture}[scale=0.95] \configtikz
\draw  node[tnode] (x2) [label=left:$x$] {$\neg w$ $3^+$ $7^-$} --
++(-90:1cm) node[blacknode] (y2) [label=left:$a$] {} -- ++(-90:1cm)
node[tnode] (w2) [label=left:$y$] {\small{$\neg w$ $7^-$}}; \draw (x2)
edge [post,bend left] node [label=right:$\frac{1}{2}$] {} (y2);

\end{tikzpicture}

%% file: rule15.tex
\begin{tikzpicture}[scale=0.95] \configtikz
\draw node[tnode] (x2) [label=left:$x$] {$3^+$ $7^-$} -- ++(-90:1cm)
node[blacknode] (y2) [label=left:$a$] {} -- ++(-90:1cm) node[tnode]
(w2) [label=left:$y$] {\small{$8^+$ $14^-$}}; \draw (x2) edge
[post,bend left] node [label=right:$\frac{3}{8}$] {} (y2);
\end{tikzpicture}

%% file: rule16.tex
\begin{tikzpicture}[scale=0.95] \configtikz
\draw node[tnode] (x2) [label=left:$x$] {$3^+$ $7^-$} -- ++(-90:1cm)
node[blacknode] (y2) [label=left:$a$] {} -- ++(-90:1cm) node[tnode]
(w2) [label=left:$y$] {\small{$15^+$}}; \draw (x2) edge [post,bend
left] node [label=right:$\frac{1}{5}$] {} (y2);
\node[texte, right=1pt] at (0.3,-1.2) {$\overline{\rm{negative}}$};
\end{tikzpicture}

%% file: rule2.tex
\begin{tikzpicture}[scale=0.95] \configtikz
\draw node[tnode] (w1) {} -- ++(180:1cm)
node[blacknode] (u1) {} -- ++(180:1cm) node[blacknode] (u2)
[label=180:$u$] {} --
++(-120:1cm) node[tnode] (w2) {\small{$3^+$ $7^-$}};
\draw (u2) -- ++(120:1cm) node[tnode] (y3) [label=180:$x$]
{\small{$3^+$ $7^-$}}; \draw (y3) edge [post,bend left] node
[label=right:$\frac{1}{10}$] {} (u2);

\end{tikzpicture}

%% file: rule3.tex
\begin{tikzpicture}[scale=0.95] \configtikz
\draw (0,0) node[tnode] (x2)
[label=left:$x$] {$8^+$ $14^-$} -- ++(-90:1cm) node[whitenode] (y2)
[label=left:$a$] {}; \draw (x2) edge [post,bend left] node
[label=right:$\frac{5}{8}$] {} (y2);
\end{tikzpicture}

%% file: rule4.tex
\begin{tikzpicture}[scale=0.95] \configtikz
\draw (0,0) node[tnode] (x2)
[label=left:$x$] {$15^+$} -- ++(-90:1cm) node[whitenode] (y2)
[label=left:$a$] {}; \draw (x2) edge [post,bend left] node
[label=right:$\frac{4}{5}$] {} (y2);
\end{tikzpicture}